\documentclass[a4paper,UKenglish,cleveref, autoref, thm-restate]{lipics-v2021}
\usepackage[T1]{fontenc}
\usepackage{graphicx}
\usepackage{tikz}
\usetikzlibrary{automata, arrows.meta, positioning,calc}
\usepackage{float}

\usepackage{amsmath}
\usepackage{amssymb}

\newcommand{\EDGE}[2]{(#1,#2)}

\title{Canadian Traveller Problems in Temporal Graphs}

\author{Thomas Bellitto}{Sorbonne Université, CNRS, LIP6, F-75005 Paris, France}{thomas.bellitto@lip6.fr}{https://orcid.org/0009-0001-5424-0742}{}

\author{Johanne Cohen}{CNRS-LISN, Université Paris-Saclay, France}
{Johanne.Cohen@universite-paris-saclay.fr}{https://orcid.org/0000-0002-9548-5260}{}

\author{Bruno Escoffier}{Sorbonne Université, CNRS, LIP6, F-75005 Paris, France \and Institut Universitaire de France}{bruno.escoffier@lip6.fr}{https://orcid.org/0000-0002-6477-8706}{}

\author{Minh-Hang Nguyen}{Université Paris Cité, CNRS, IRIF, F-75013, Paris, France}{minh-hang.nguyen@irif.fr}{https://0009-0008-2391-029X}{}

\author{Mikaël Rabie}{Université Paris Cité, CNRS, IRIF, F-75013, Paris, France}{mikael.rabie@irif.fr}{https://orcid.org/0000-0001-6782-7625}{}

\authorrunning{T. Bellitto et al.} %TODO mandatory. First: Use abbreviated first/middle names. Second (only in severe cases): Use first author plus 'et al.'

\Copyright{Thomas Bellitto, Johanne Cohen, Bruno Escoffier, Minh Hang NGuyen, Mikael Rabie} %TODO mandatory, please use full first names. LIPIcs license is "CC-BY";  http://creativecommons.org/licenses/by/3.0/

 %TODO mandatory: Please choose ACM 2012 classifications from https://dl.acm.org/ccs/ccs_flat.cfm 

 %TODO mandatory; please add comma-separated list of keywords

\category{} %optional, e.g. invited paper

\relatedversion{} %optional, e.g. full version hosted on arXiv, HAL, or other respository/website
\acknowledgements{This work was supported by the French ANR project TEMPOGRAL (ANR-22-CE48-0001). The fourth author has received funding from the European Union's Horizon 2020 research and innovation program under the Marie Skłodowska-Curie grant agreement No 945332.
}%optional

\EventEditors{}
\EventNoEds{0}
\EventLongTitle{}
\EventShortTitle{}
\EventAcronym{}
\EventYear{}
\EventDate{}
\EventLocation{}
\EventLogo{}
\SeriesVolume{}
\ArticleNo{}

\begin{document}

\maketitle              % typeset the header of the contribution

\begin{abstract}
We focus on the Canadian Traveller Problem, where a traveller aims to travel from $s$ to $t$ on a network  with the minimum cost, 
with up to $k$ edges potentially blocked. These edges remain hidden  until the traveller visits one of their endpoints. 

We have investigated this problem in temporal graphs where edges are  accessible only at specific times. We explore Three classical variants of the shortest path problem:  The \emph{Latest departure path}, the  \emph{Earliest arrival path}, \emph{and the Shortest duration path}.

%This problem has yet to be explored in temporal graphs, where three classical variants of the shortest path problem can be considered:  The Latest departure path, the Earliest arrival path, and the Shortest duration path.

%This paper formalises this problem as a positional game in two-player graphs. First, we highlight the distinction between the temporal and non-temporal versions of the problem. Then, we focus on providing efficient exact algorithms for directed acyclic graphs from source to destination using dynamic programming. Additionally, we explore the complexity of finding a winning strategy for the traveller, proven to be PSPACE-complete. \footnote{Do we want to highlight in the abstract that we prove that in the general case, it is $W[1]$ on $k$?}

This paper formalises the Canadian Traveller problem as a positional two-player game on graphs. We consider two variants depending on whether an edge is blocked. In the \emph{locally-informed} variant, the traveller learns if an edge is blocked upon reaching one of its endpoints,  while in the \emph{uninformed} variant, they discover this only when the edge is supposed to appear. We provide a polynomial algorithm for each shortest path variant in the uninformed case. This algorithm also solves the case of directed acyclic non-temporal graphs. 

In the locally-informed case, we prove that finding a winning strategy is PSPACE-complete. Moreover, we establish that the problem is polynomial-time solvable when $k=1$ but NP-hard for $k\geq 2$. 

Additionally, we show that the standard (non-temporal) Canadian Traveller Problem is NP-hard when there are $k\geq 4$ blocked edges, which is, to the best of our knowledge, the first hardness result for CTP for a constant number of blocked edges. 

\end{abstract}

\newpage

\section{Introduction}

\subsection{Motivation}

  A Canadian traveller wants to travel from one city to another but may discover that a road they intended to take is blocked (by snowfalls, for example). They must plan robust itineraries to account for possible road closures, ensuring backup routes are efficient.  In this paper, we extend this problem to temporal graphs,  where edges between vertices exist only at specific timeframes, such as a train network. We study several variants depending on how the problem is transposed to temporal graphs. Upon arriving at a train station, it is sometimes known in advance which trains are missing, providing early information to adapt plans accordingly. At other times, it is learned at the last minute that the planned train has been cancelled.

  Temporal graphs introduce complexities beyond static graphs in algorithmic problems. For example, finding a minimum set of edges that connects all the vertices of a graph is an easy problem known as a Minimum Spanning Tree in the static case. Still, temporal spanners raise many very challenging problems (see~\cite{Spanners2,Spanners1} for example). Since the Canadian Traveller Problem is already PSPACE-complete, studying it in temporal graphs could seem hopeless. However, one of the reasons why the Canadian Traveller Problem is so difficult is that as the traveller moves in the graph, they have to remember which roads are blocked and which ones are not to choose their next move. We refer the reader to Section~\ref{subsection:preliminaries} for a more detailed presentation of the Canadian Traveller Problem and an illustration of why remembering which edges are blocked is important. On the other hand, in a temporal graph, edges appear and disappear at every timeframe, and once an edge is gone, it does not matter anymore whether it was usable or not in the first place.  The previous remark suggests that specific variants of our problem are solvable more efficiently in temporal graphs than in static graphs.

\subsection{Related Work}

Papadimitriou and Yannakakis introduced the Canadian Traveller Problem (CTP) in 1991~\cite{papadimitriou1991shortest}. Their paper aims to find a path that minimises the sum of the weight of its edges. The challenge is that the traveller initially only knows an interval of possible values for its weight for each edge. They learn the exact weight of an edge when they reach an adjacent neighbour. If an edge is heavier than the traveller hoped, they can still choose to switch to another route. The objective is to design a strategy minimising the worst possible ratio between the chosen and optimal routes. Blocking edges can be simulated by assigning them arbitrarily large weights.

From a complexity perspective, Papadimitriou and Yannakakis~\cite{papadimitriou1991shortest} showed that determining a strategy ensuring a specified competitive ratio is PSPACE-complete when the number of potentially blocked edges is not fixed. Bar-Noy and Schieber~\cite{bar1991canadian} then studied a problem closer to ours, where the weight of the edges is known in advance, but edges can be blocked. They formulated the problem with a bounded number of blocked edges as a game between a traveller and an adversary who can block edges during the game. They show that the problem is PSPACE-complete when the number of blocked edges $k$ is a parameter and that the problem is polynomial for fixed $k$. We refer the reader to Appendix~\ref{app:difference} for discussing the relation between~\cite{bar1991canadian} and our article.

Despite the negative results, the CTP is also studied under stochastic settings~\cite{nikolova2008route,polychronopoulos1996stochastic} when edge  weights    are treated as random variables with specified distributions. In particular,  Nikolova and Karger~\cite{nikolova2008route} address  this case on directed acyclic and path-disjoint graphs, providing a polynomial algorithm to optimise the expected travel time.

The strategy for a traveller in CTP can be viewed as an \emph{online algorithm} because the traveller lacks prior knowledge of blocked edges and discovers them them upon reaching their endpoints. In classical competitive analysis, many studies explicitly consider the impact of the number $k$ of blocked edges on the analysis. In 2009, Westphal~\cite{westphal2008note} showed that no deterministic online algorithm can achieve a competitive ratio smaller than $2k+1$. They  also provided a simple online algorithm to reach this lower bound. Moreover, Westphal~\cite{westphal2008note}  showed that no randomised online algorithm can
no randomized online algorithm can achieve better than  $(k + 1)$-competitive ratio. Bender and Westphal~\cite{bender2015optimal} showed that the bound is tight by constructing a randomised online algorithm for this case. Demaine, Huang, Liao and Sadakane~\cite{demaine2021approximating} introduced a randomised algorithm with polynomial time complexity, that  outperforms the best deterministic algorithms when there are at least two
blockages, and surpasses the lower bound of $2k + 1$ by an $o(1$) factor.

Dynamic networks have been considered under several definitions and names. In particular, Temporal Graphs are graphs coupled with a discrete (often finite) time. At each time, a subset of the edges is present (sometimes, with some travel time). They serve as models for real-world systems such as transportation networks~\cite{holme2012temporal}. The notion of paths and connectivity have been extensively studied; see, for example~\cite{awerbuch1984efficient,DBLP:journals/jcss/KempeKK02}.  Questions typically revolve around the earliest arrival time, latest departure time, or finding shortest paths in terms of travel time or distance (number of edges).
%From a source to a target, the usual questions are: What is the earliest time we can arrive, when can we leave at the latest, or what is the shortest path, either in terms of travel time or distance (\textit{i.e.}, number of edges).
In~\cite{DBLP:conf/stacs/FuchsleMNR22}, the authors consider robust itineraries in the context of possible delays on edges. A stochastic method given a distribution on the possible delays has been studied in~\cite{DBLP:conf/atmos/DibbeltSW14}.

\section{Definitions and Contributions}

\subsection{Preliminaries}\label{subsection:preliminaries}

Let us first define the traditional Canadian Traveller Problem and Temporal Graphs before bringing the definitions of our problems:

\begin{definition}\label{def:staticctp}(Static Canadian Traveller Problem (CTP))
Given a graph $G$, a multiplicity function on the edges $mult: E\to \mathbb{N}^*$, a distance function on the edges $d: E\to \mathbb{N}^*$, two specific vertices $s$ and $t$, an integer $k$, a maximal travel time $T$ and two players named \emph{Traveller} and \emph{Blocker}, we consider the following game, where the two players play as follows:
\begin{itemize}
    \item (Blocker plays) If Traveller visits for the first time vertex $v$, then Blocker decides (and tells Traveller) which edges incident to $v$ are blocked and which are not. Each edge $e$ has $mult(e)$ copies, and Blocker can block up to $mult(e)$ of them. The graph's total number of blocked edges cannot exceed $k$. If Traveller has already visited $v$ before, then Blocker does not play. 
    \item (Traveller plays) Traveller, currently at $v$, chooses a non-blocked edge $(v,u)$ to follow, spending time $d(u,v)$.
\end{itemize}
Traveller is initially at $s$ at $0$  and wins the game if they reach $t$ with total travel time at most $T$.
\end{definition}

In this paper, we will omit $mult$ and explicitly declare how many copies of each edge appear. {We will refer to $d(e)$  as the distance or length of edge $e$}.
Following this definition, the status of an edge (blocked/non-blocked) remains unchanged:  Blocker decides if an edge is blocked or not the first time the Traveller reaches one of its endpoints, and  Blocker cannot change this status later.

We consider the example depicted in Figure~\ref{figintroCTP}. Under the assumption that Blocker can block up to two edges and that Traveller only learns it when they reach an adjacent vertex, {Traveller can go} from $s$ to $t$ with a budget of $17$. Indeed, right from the start, Traveller knows which edges $\EDGE{s}{a}$, $\EDGE{s}{b}$ and $\EDGE{s}{c}$ are blocked. If the edge $\EDGE{s}{b}$ or $\EDGE{s}{c}$  is blocked and $\EDGE{s}{a}$ is not, Traveller can safely go to $a$ and then go to $t$ directly or through $d$. If both $\EDGE{s}{a}$ and $\EDGE{s}{b}$ are blocked, the path $s\to c\to a\to t$ is available and within the budget. If no edge adjacent to $s$ is blocked, the shortest path from $s$ to $t$ is $s\to a\to t$, but trying to use it would be unwise since $\EDGE{a}{t}$ and 
$\EDGE{a}{d}$ could be blocked, and Traveller would have to go back to $s$ and cannot win from there with a remaining budget of $5$. The right move is actually to go to $b$. From there, if $\EDGE{b}{t}$  is available, Traveller can use it and win. Let us assume that Traveller is on vertex $b$ with a remaining budget of $13$ and that $\EDGE{b}{t}$  is blocked. If the second blocked edge is $\EDGE{e}{t}$, going to $e$ would lose the game as the final path of Traveller would be $s\to b\to e\to b\to s\to a\to t$, which costs $20$. However, if Traveller returns to $s$ and realises that $\EDGE{s}{a}$ is blocked, their final route would be $s\to b\to s\to b\to e\to t$, which costs $18$. The trick here is that even though Traveller did not want to use $\EDGE{s}{a}$ on their first move, they could still see if it was blocked and must remember it to choose their next move. If $\EDGE{s}{a}$ was blocked, they know no other edge is and can safely reach $t$ through $e$. On the other hand, if $\EDGE{s}{a}$ was usable, they should backtrack to $s$ and use $\EDGE{s}{a}$. They are now standing on $a$ with a remaining budget of $3$, and Blocker only has one edge left to block. In the worst case, Blocker blocks $\EDGE{a}{t}$, and Traveller's final route would be $s\to b\to s\to a\to d\to t$, whose total cost is $17$.

\begin{figure}[!h]
\centering
\includegraphics[]{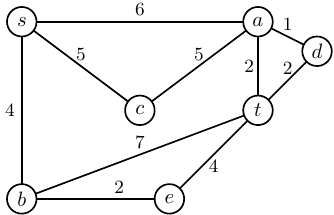}
\caption{An instance of the Canadian Traveller Problem with weighted edges.}
\label{figintroCTP}
\end{figure}

%We note that~\cite{bar1991canadian} seems to consider\footnote{The model is not precisely defined in the article.} a slightly different model where Blocker can first decide that an edge in non-blocked (the first time Traveller visits one of its endpoints) and then turns it into a blocked edge (if Traveller goes a second time to one of its endpoints).\\

Let us now define temporal graphs.

\begin{definition}(Temporal graph)
    In a temporal graph $G=(V, E)$, we are given a set $V$ of vertices and a set $E$ of time edges. A \emph{time edge} $(u,v,\tau,d)$ is an edge $\{u,v\}$ at time $
\tau\in \mathbb{N}$ taking $d\in\mathbb{N}\setminus\{0\}$ units of time to cross this edge. Note, as $G$ is undirected, that $(u,v,\tau,d)=(v,u,\tau,d)$ is a single edge that can be described from both directions. If $d=1$, we will omit it by default in the description (in particular, in Section~\ref{sec:linf}).
\end{definition}
{For a time edge $(u,v,\tau,d)$, $d$ corresponds to the travel time of this edge and will also be referred to as the length of the time edge.} In the {following}, we will {often} say edge instead of time edge when considering a temporal graph.
%\begin{definition}(Temporal walk)
In a temporal graph $G=(V,E)$, a \emph{temporal walk} is a sequence $(v_0,e_0,v_1,e_1,\dots,e_{p-1},v_p)$ with $v_i\in V$ and $e_i=(v_i,v_{i+1},\tau_i,d_i)\in E$, such that $\tau_i+d_i\leq \tau_{i+1}$. %\end{definition} 
In such a walk, we arrive at $v_{i+1}$ at $\tau_{i}+d_i$ (using the edge from $v_i$) and leave $v_{i+1}$ at $\tau_{i+1}$. We {assume that an edge can appear} several times. In that case, we will specify how many copies we have in the graph.

\subsection{Temporal CTP}
In a temporal CTP, a traveller travels (i.e., follows temporal walks) through a temporal graph $G$. Some edges may be blocked. It is important to note that Blocker can block is time edges. {For example, while one time edge  $(u,v,\tau,d)$ might be blocked,  another time edge $(u,v,\tau', d')$ between the same vertices may remain unblocked.
This distinction is significant because} if Blocker could block all  time edges between two vertices {simultaneously}, the resulting problem would trivially be a generalisation of the static problem (which is the case of a temporal graph where the same edges are present at every time step). The complexity results follow immediately.

Traveller is {made} aware of a blocked edge upon reaching one of its endpoints. We define two variants of the problem based on when Traveller discovers a blocked edge:

\begin{itemize}
    \item \emph{Locally informed model:} 
   Traveller is {made} aware of all the blocked edges incident to $u$ when they {arrive} at node $u$  (including edges $(u,v,\tau,d)$ for $\tau$ greater than the \textcolor{blue}{arrival} time). 
    \item \emph{ Uninformed model}: Traveller is {made} aware of whether an edge $(u, v, \tau, d)$  is blocked only when located at node $u$ (or at node $v$) at time $\tau$.
\end{itemize} 

Let us now formally define these variants. As mentioned in the introduction, CTP can be {viewed} as a two-player game between a traveller, aiming to reach their destination, and a blocker aiming {to} block Traveller.

\begin{definition}(Locally informed game)
Given a temporal graph $G$, two specific vertices $s$ and $t$, an integer $k$, and two players named \emph{Traveller} and 
\emph{Blocker}, {the game proceeds as follows:} %we consider the following game, where the two players play as follows:
\begin{itemize}
    \item (Blocker plays) {When} Traveller visits for the first time vertex $v$, Blocker decides and tells Traveller {about} which edges incident to $v$ are blocked and which are not. The total number of blocked edges in the graph  cannot exceed $k$. If Traveller has {previously}  visited $v$, then Blocker does not play. 
    \item (Traveller plays) Traveller, currently at $v$ at time $\tau$, chooses a non-blocked edge $(v,u,\tau',d)$ to {traverse} (with $\tau'\geq \tau)$, i.e., they move from $v$ to $u$ at time $\tau'$, arriving at $u$ at time $\tau'+d$.
\end{itemize}
Traveller is initially at vertex $s$ at $0$ and  {aims to reach vertex $t$ to win}. %wins the game if they reach $t$.
\end{definition}

We note that, here also, the set of blocked/non-blocked edges does not evolve: once Blocker decides whether $(u,v,\tau,d)$ is blocked, they cannot change this decision later. 

Also, {it is important} to note that the problems of deciding if $t$ is reachable before a given deadline $T$ or 
{or determining if $t$ is reachable at all are easily reducible by simply removing all edges occurring after time $T$.}

%at all are easily reducible to each other by simply removing all edges after time $T$.

\begin{definition}(Uninformed game)
An \emph{uninformed game} is defined similarly to a locally informed {game with the following modification:}  Traveller needs to be at vertex $v$ at time $\tau$ for the Blocker to decide (and tell them) {whether} the edge $(v,u,\tau,d)$ is blocked or not.
%Given a temporal graph $G$, two specific vertices $s$ and $t$, an integer $k$, and two players named Traveller and Blocker, we consider the following game, where the player plays alternatively:
%\begin{itemize}
%    \item (Blocker plays) If Traveller is at vertex $v$ at time $\tau$, then Blocker must decide (and tells Traveller) which edges $(v,u,\tau,d)$ are blocked and which are not. The graph's total number of blocked edges cannot exceed $k$.
%    \item (Traveller plays) Traveller, currently at $v$ at time $\tau$, either chooses a non-blocked edge $(v,u,\tau,d)$ to follow, or stays at $v$ (till the next time $\tau'$ such that there is an edge $(v,u,\tau',d)$ in $G$, that Blocker may decide (at time $\tau'$) to block or not).
%\end{itemize}
%Traveller is initially at $s$ at $\tau=0$ and wins the game if they reach %$t$.
\end{definition}
Note that this version gives more power to Blocker (as they reveal less information to Traveller). We now introduce the different decision problems in those games:%As time strictly increases, Blocker player has possibility to cut edges at each step (as long as less than $k$ edges were cut).

\begin{definition}(Li-TCTP, U-TCTP)
The \emph{Locally informed Temporal Canadian Traveller} Problem (resp. \emph{Uninformed Temporal Canadian Traveller Problem}), Li-TCTP for short (resp. U-TCTP for short) asks whether there exists a winning strategy for Traveller in the corresponding Locally informed game (resp. Uninformed game).
\end{definition}

Problems Li-TCTP and  U-TCTP  are reachability problems. We can also formulate optimisation problems. Let us say that Traveller has a $(T_1,T_2)$-winning strategy if they can ensure reaching the destination $t$ at time $T_2$ or sooner, by departing no sooner than $T_1$; 
\begin{itemize}
    \item \emph{(Earliest Arrival)} We aim to find the smallest $T_2$ and a corresponding strategy such that Traveller has a $(0, T_2)$-winning strategy.
    \item \emph{(Latest Departure)} We aim to find the largest $T_1$ (and a corresponding strategy) such that Traveller has a $(T_1,\infty)$-winning strategy.
    \item \emph{(Shortest Path)} We want to find $(T_1, T_2)$ with the smallest difference $T_2-T_1$ and a corresponding strategy such that Traveller has a $(T_1, T_2)$-winning strategy.
\end{itemize}
Note that if one can determine whether Traveller has a $(T_1, T_2)$-winning strategy, they can solve all three optimisation problems.\\

%\textcolor{blue}{Mikaël: I do not know where to put that}

\subsection{Static versus Temporal Case.}
In the static version, reachability with $k$ blocked edges is possible if and only if  $s$ and $t$ are $k+1$-connected. If Blocker wins, it implies that the blocked edges necessarily form a cut between $s$ and $t$, and vice versa. However, this approach can no longer work in the temporal case. As  time augments, connectivity naturally decreases, even without cutting edges. It is not enough to ensure reachability from $s$ to $t$ with $k$ blocked edges when the underlying graph of the temporal graph is $k+1$ connected.

To show the difference between Locally Informed and Uninformed, let us consider the example in Figure~\ref{fig:example-simple}, where all edges have length~$1$. First, observe that if there are $(k + 1)$ parallel edges $(u, v,\tau,d )$ ($k$ being the maximal number of blocked edges), then at least one of them remains accessible, meaning
that Traveller can always rely on taking one of these edges. In the remaining,  we refer to an edge $(u,v,\tau,d)$ as a \emph{forced edge}   to mean that there are $(k+1)$ copies of it. Hence, Blocker cannot block a forced edge.

In U-TCTP (Uninformed case), it is easy to see that there is no winning strategy for Traveller: they go from $s$ to $v_0$. At time $1$, they must either wait at $v_0$ or go to $v_1$. In the first case, Blocker blocks edge $(v_0,v_2,2,1)$. In the other case, Blocker blocks the $k$ parallel edges $(v_1,t,2,1)$. 

In Li-TCTP (Locally informed case), there is a winning strategy: Traveller goes to $v_0$. There, they know %(when arriving at $v_0$, i.e., at $\tau=1$) 
if $(v_0,v_2,2,1)$ is blocked or not. If it is not, Traveller takes this edge and reaches $t$ (as $(v_2,t,3,1)$ is forced). If $(v_0,v_2,2,1)$ is blocked, Traveller goes to $v_1$, and then to $t$: Blocker cannot block the $k$ parallel edges $(v_1,t,2,1)$ as their remaining budget (corresponding to the remaining number of edges that can still be blocked) is $k-1$ (they have already blocked one edge).

{This example highlights how the timing of when Traveller learns about blocked edges significantly impacts their ability to strategize effectively in temporal graphs.}

\begin{figure}[H]
\begin{center}
\begin{tikzpicture}
\begin{scope}[every node/.style={circle,thick,draw}]
    \node (S) at (0,0) {$s$};
    \node (V0) at (2,0) {$v_0$};
    \node (V1) at (4,0) {$v_1$};
    \node (V2) at (2,1.5) {$v_2$};
    \node (t) at (4,1.5) {$t$} ;
\end{scope}

%% Forced edges
\begin{scope}[>={Stealth[black]},
              every node/.style={circle},
              every edge/.style={draw=black, dotted}]
    \draw [-,ultra thick,double] (S) to node[above,pos=0.5]{0}  (V0);
    \draw [-,ultra thick,double] (V0) to node[above,pos=0.5]{1}  (V1);
    \draw [-,ultra thick,double] (V2) to node[above,pos=0.5]{$3$}  (t);
\end{scope}

%% k copies
\begin{scope}[>={Stealth[black]},
              every node/.style={circle},
              every edge/.style={draw=black, thick}]
    \draw [-,ultra thick] (V1)    to 
node[right,pos=0.5]{2} (t);
\end{scope}

%% normal
\begin{scope}[>={Stealth[black]},
              every node/.style={circle},
              every edge/.style={draw=blue, thick}]
      \draw [-] (V0) to node[right,pos=0.50]{2} (V2);
\end{scope}

\end{tikzpicture}
\caption{{Example of a} temporal graph. Thick (resp. doubly thick) edges represent $k$ parallel (resp. $(k+1)$-parallel) edges. The number of the edge corresponds to the time step in which it is present. All edges have length one.}
  \label{fig:example-simple}
\end{center}
\end{figure}
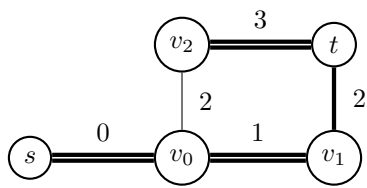

\subsection{Contributions}

We focus on the computational complexity of solving U-TCTP and Li-TCTP, highlighting a significant difference between these two problems. 

In Section~\ref{sec:uninf}, we  show that U-TCTP is solvable in time polynomial in $n=|V|$ and $k$. This result, valid also for the three optimisation versions, is obtained by reducing our problem to a CTP in static {\it directed acyclic} graphs, a problem that we show how to solve efficiently. 
{We leverage classical Dynamic Programming (DP) techniques to devise a polynomial-time algorithm for U-TCTP, despite its static version being computationally hard.}

%The solution of U-TCTP, obtained using classical Dynamic Programming (DP) techniques, is interesting as it provides a polytime algorithm for a temporal problem whose static version is hard. 

We then tackle Li-TCTP in Section~\ref{sec:linf}. We show that this problem (even in its reachability version) is PSPACE-hard. We then show how to solve the case $k=1$ in polynomial time (deferred to Appendix~\ref{app:k=1}). %We leave the complexity of Li-TCTP parameterised by $k$ as an open problem.
%We then tackle (Li-TCTP) in Section~\ref{sec:linf}. We show that this problem (even in its reachability version) is PSPACE-hard. We then show how to solve the case $k=1$ in polynomial time (deferred to Appendix~\ref. We leave the complexity of (Li-TCTP) parameterised by $k$ as an open problem.

%We also hint why solving the problem for larger $k$ is harder than considering all possible subsets of size at most $k$ of blocked edges as it was (wrongly) suggested in~\cite{bar1991canadian}. 

In Section~\ref{sec:static}, we prove that Li-TCTP is NP-hard for $k=2$. We also consider the static case (see Definition~\ref{def:staticctp}) and show that static CTP becomes NP-hard when there are $k=4$ blocked edges. To the best of our knowledge, this result is the first hardness result for (static) CTP for a constant  $k$ (we refer the reader to Appendix~\ref{app:difference} for a discussion with the results in~\cite{bar1991canadian}). 
We conjecture that static CTP is NP-hard already for  $k=2$ but leave this as an open question.

%\input{example1}

%\input{universal-example}
 
%M3) Traveller discover the missing edge $(v,u,t,d)$ at time $t$ regardless which node he stays at that time.

\section{Uninformed Temporal CTP}\label{sec:uninf}

\subsection{CTP for directed acyclic static graphs}

This section focuses on the Canadian Traveller Decision Problem in a directed acyclic static graph (DAG). CTP in acyclic graphs has already been studied in~\cite{nikolova2008route}. However, they consider probabilistic distribution for the weight on edges and optimise the expected time. Each edge can be affected in their work when our algorithm needs to ponder how many remaining edges can be blocked. Our approach addresses a deterministic scenario where the primary concern is the blocking of edges. Specifically, we show that this problem in DAGs can be solved in polynomial time even when up to kk edges can be blocked.

Given a DAG $G = (V, E)$, with a source node $s$ and a target node $t$, a traveller wants to go from $s$ to $t$ when Blocker can block at most $k$ arcs. Traveller can only discover the blocked arc when they reach the head of that blocked arc.

\begin{theorem}\label{lem:acyclic1}
    %The Canadian Traveller decision Problem in a weighted DAG $G$ can be solved in polynomial time in $n,m$ and $k$. 
    The Canadian Traveller Decision Problem in a weighted DAG $G$ can be solved in time $O(k^2(m+n))$. 
\end{theorem}

Briefly, each node $v$ maintains a set of $k+1$ values, denoted as $\pi_0,\ldots,\pi_k$. Informally, $\pi_i$ represents the latest arrival time from which Traveller can guarantee to go from $v$ to $t$ under {the condition} when Blocker can block $i$ arcs between node $v$ and $t$.

{Since $G$ is a DAG, we compute a topological order on the vertices, of which we can assume that $s$ is the first one and $t$ is the last one. We proceed in reverse topological order.}

{We first consider the case of a node $v$ which has only $t$ as outneighbor. $\pi_i(v)$} is the maximum cost Traveller has to pay to reach $t$ {from $v$} when Blocker can block at most $i$ arcs connecting $v$ and $t$. Assume that there are $m$ arcs from $v$ to $t$, then $\pi_i$ is the $(i+1)$-th minimum weight of an arc from $v$ to $t$ since Blocker can block $i$ arcs with smallest weights. We denote by $\text{i-th} \min$ the $i$-th smallest value of a set. If $i\geq m$ then $\pi_i(v) = \infty$. Otherwise,
$$\pi_i(v) = \text{(i+1)-th} \min d(v,t) $$
{Otherwise, let $N^+(v)$ be a set of all out-neighbours of $v$.} The values, $\pi_0,\ldots, \pi_k$ for each node are computed inductively as follow:

$$ \pi_0(v) = \min_{u \in N^+(v)} \left(\pi_0(u)+d(v,u) \right)
$$

$$ \pi_i(v) = \max \left\{ \begin{array}{l}
    \displaystyle    \min_{u \in N^+(v)} \left(\pi_i(u)+d(v,u) \right)\\
     \displaystyle    \text{second} \min_{u \in N^+(v)} \left(\pi_{i-1}(u) +d(v,u) \right)\\
        \ldots\\
     \displaystyle    \text{i-th} \min_{u \in N^+(v)} \left(\pi_{1}(u)+d(v,u)\right)\\
    \displaystyle     \text{(i+1)-th} \min_{u \in N^+(v)} \left(\pi_{0}(u)+d(v,u)\right)\\
    \end{array} \right.
$$
For each $i$, values $\pi_i$ for all vertices can be computed in time $O(k(m+n))$. Therefore, values $\pi_i(v)$, for $i=0,\ldots,k$ and $v\in V$, can be computed in time $O(k^2(m+n))$.

We now consider the CTP for directed acyclic graphs a two-player game.

\begin{lemma}\label{lem:acyclic2}
For the Canadian Traveller Decision Problem in a weighted {DAG} $G$,   Traveller has a winning strategy to reach $t$ with arrival time at most $T$ if and only if $\pi_k(s) \leq T$. %\textcolor{blue}{Comment: I removed the sentence on Blocker, which was useless as we use if and only if}%Blocker has a winning strategy, i.e., a strategy to prevent Traveller from reaching $t$ within budget $T$, if $\pi_k(s) > T$.
\end{lemma}

The proof of this Lemma is in Appendix \ref{app:lemacyclic}.

%Theorem \ref{lem:acyclic1} is proved by Lemma \ref{lem:acyclic2}.

\subsection{CTP for Temporal graph in the uninformed case}

\begin{theorem}\label{th:U-TCTP}
    Uninformed-Temporal Canadian Traveller Problem can be solved in polynomial time.
\end{theorem}
\begin{proof}
Proving this theorem requires transforming the temporal CTP into a CTP on a static weighted DAG. Given a temporal graph ${G} = (V, E)$, we construct a weighted directed acyclic graph $G' = (V', E')$ corresponding to $G$. This construction is known as the static expansion of a temporal graph~\cite{mertzios2013temporal}.
%, where $E$ is a set of time edges of $G$, and $E'$ is a set of weighted edges of $G'$. 
%\textcolor{blue}{Bruno: We can remove what is after ``where'', no?}

For each edge $(v,u,\tau,d) \in E$, we create two vertices $(v,\tau)$ and $(u,\tau+d)$ in $V'$. We create one universal target node $\texttt{target} \in V'$. 
\begin{itemize}
    \item If $(v,u,\tau,d) \in E$, then $(v,\tau)(u,\tau+d) \in E'$, and $d((v,\tau),(u,\tau+d)) = d $. We call $(v,\tau)(u,\tau+d)$ the corresponding edge of $(v,u,\tau,d)$.
    \item For all $v \in V, (v,\tau) \in V'$, and the minimum $\tau'>\tau$ such that $(v,\tau') \in V'$, let {$((v,\tau),(v,\tau'))\in E'$ and} $d((v,\tau),(v,\tau')) = \tau' -\tau$.
    \item For every $(t,\tau) \in V'$, connect node $(t,\tau)$ to $\texttt{target}$ by $k+1$ parallel edges, let $d((t,\tau),\texttt{target}) = 0$.
    
\end{itemize}
Since $G$ is a temporal graph, $G'$ is a directed acyclic graph. In $G'$, source node is $(s,0)$ and target node is $\texttt{target}$.

%\begin{lemma}\label{lem:U-TCTPtoCTP}
    Then, there exists a winning strategy for Traveller in U-TCTP in $G$ if and only if there exists a winning strategy for Traveller in Canadian Traveller Decision Problem with reachable condition in directed acyclic graph $G'$. 
\end{proof}

Note that to solve U-TCTP for three optimisation versions - Latest departure, Early arrival, and Shortest time travelling - {we proceed as follows:} For each $T_1 \leq T_2 \leq T$, {where $T$ is the maximum value of $\tau$ over time edges $(u,v,\tau,d)$ (also known as the lifespan of the graph),} we consider a subset of time edges of $G$, which are presented between time $T_1$ and $T_2$. Then, we construct a corresponding directed acyclic graph $G'(T_1, T_2)$.

%---------------------------------------------------------------------------------------
\newcommand{\FF}{{L +1}}
\newcommand{\FIN}{{L+5 }}
\newcommand{\FINmoinsUN}{{L+4}}
\newcommand{\FINmoinsDEUX}{{L + 3}}
\newcommand{\FINmoinsTROIS}{{L +2}}
\newcommand{\FINmoinsQUATRE}{{\FF}}

\newcommand{\impairUN}{{7\lfloor \frac{i}2 \rfloor+1}}
\newcommand{\impairDEUX}{{7\lfloor \frac{i}2 \rfloor+2}}
\newcommand{\impairTROIS}{{7\lfloor \frac{i}2 \rfloor+3}}

  \newcommand{\pairUN}{{7\lfloor \frac{i-1}2 \rfloor+3}}
  \newcommand{\pairDEUX}{{7\lfloor \frac{i-1}2 \rfloor+4}}
  \newcommand{\pairTROIS}{{7\lfloor \frac{i-1}2 \rfloor+5}}
  \newcommand{\pairQUATRE}{{7\lfloor \frac{i-1}2 \rfloor+6}}
  \newcommand{\pairCINQ}{{7\lfloor \frac{i}2 \rfloor}}

\newcommand{\KK}{{2m-1}}
\newcommand{\K}{{k}}
\newcommand{\Kvalue}{{2m+\frac{n}2}}
\newcommand{\KmoinsUN}{{k-1}}
\newcommand{\KmoinsDEUX}{{k-2}}

\section{Locally Informed Temporal CTP}\label{sec:linf}
%\subsection{PSPACE-hardness} 
\begin{theorem}\label{th:pspace}
  Li-TCTP is PSPACE-complete, even in instances where {each edge has length $1$ and is present at only a one-time step. This complexity applies to both directed and undirected (temporal) graphs.}
\end{theorem}
\begin{proof}
	This problem constitutes a two-player game {for which membership in PSPACE is easily shown using standard techniques (Lemma 2.2 of~\cite{schaefer1978complexity}).} 
 %for which it is easy to show membership in PSPACE using standard techniques (Lemma 2.2 of~\cite{schaefer1978complexity}).
 Now, we prove completeness. The reduction is from the quantified {3-}satisfiability problem {which is PSPACE-complete}~\cite{stockmeyer1973word}:

\begin{definition} (Quantified 3-Sat):
\noindent{\bf Given} Set $X = \{x_{1}, \dots,  x_{n}\}$ of $n$ variables ($n$ even) and a quantified Boolean formula
$$F = (\exists x_{1})(\forall x_{2})(\exists x_{3}) \dots (\forall  x_{n}) \phi(x_{1}, \dots,  x_{n})$$
where $\phi$ is a propositional formula over $X$ in 3-CNF (i.e., in conjunctive normal form with exactly three literals per clause).

 \noindent{\bf Question} : Is $F$ true?
\end{definition}

Let $C_{1},\dots ,C_{m}$ denote the clauses that make up $\phi$: 
$\phi(x_{1}, \dots,  x_{n}) =C_{1}\land C_{2} \dots \land  C_{m}$, where each $C_{i}$ , $1\leq i \leq m$.
 contains exactly three literals.

 Given an instance $\mathcal Q$ of Quantified 3-Sat, we construct an instance $\mathcal S$ of a Locally informed game by creating a temporal graph $G = (V, E)$ as follows. First, Blocker can block $k=\Kvalue$ %\textcolor{blue}{$2m+\frac{n}{2}$} 
 edges.

In the temporal graph $G$, all edges will have {length 1}. For clarity, we will omit to mention the {length} when referring to a temporal edge. {$G$ has:}
 \begin{itemize}
 \item a start vertex $s$,  a target vertex $t$,   and a vertex $w$
  \item vertices $v_{i}$, $x_{i}$, $\overline{x_i}$, $a_i$, and  $z_{i}$,  for each integer $i$, $1\leq i\leq n$,     (one per boolean variable) 
  %\item vertices $a_{i}$,  for each  integer $i$, $1\leq i\leq n+1$,
  \item vertices $b_{i}$,  for each  even integer $i$, $1\leq i\leq n+1$,
  \item vertex $c_{j}$,    for each integer $j$, $1\leq j\leq m$,   (one per clause) 
  \item {two} distinct vertices $v_{n+1}$, $a_{n+1}$.
\end{itemize}

Now, we describe how the vertices are connected {(see Figures~\ref{fig:PSPACE:variable} and~\ref{fig:PSPACE:clause} for illustrations)}. The set of edges can be divided into three parts. The first part corresponds to edges associated with variables of odd index, the second part to edges associated with variables of even index, and the third part corresponds to edges associated with clauses. {We define $L=7\frac n2$, a constant that we will often use, corresponding to the time at which Traveller can reach $v_{n+1}$.}

\begin{figure}[ht]
\begin{center}

\begin{tikzpicture}[scale=1]
\begin{scope}[every node/.style={circle,thick,draw}]
    \node (S) at (0,5) {$s$};
    \node (V0) at (1.5,5) {$v_1$};
    \node (A0) at (1.5,3) {$a_1$}; % AJOUT DE JOHANNE
    \node (X1) at (3,6) {$x_1$};
    \node (X1B) at (3,4) {$\overline{x_1}$};
    \node (z1) at (2.5,8) {$z_{1}$} ;  % AJOUT DE JOHANNE

    \node (V1) at (4.5,5) {$v_2$};
    \node (A1) at (4.5,3) {$a_2$}; % AJOUT DE JOHANNE

     \node (B2) at (4.5,7) {$b_2$}; % AJOUT DE JOHANNE
     \node (t4) at (6,1) {$t$} ;  % AJOUT DE JOHANNE
    \node (t1) at (6,10) {$t$} ;
 
    \node (X2) at (6,6) {$x_2$};
    \node (X2B) at (6,4) {$\overline{x_2}$};
    \node (V2) at (7.5,5) {$v_3$};  % AJOUT DE JOHANNE
     \node (A2) at (7.5,3) {$a_3$}; % AJOUT DE JOHANNE
     \node (z2) at (7,8) {$z_{2}$} ;  % AJOUT DE JOHANNE

    \node (Xn) at (10,6) {$x_n$};
    \node (XnB) at (10,4) {$\overline{x_n}$};
    \node (zn) at (9,8) {$z_{n}$} ;  % AJOUT DE JOHANNE

    \node (Vn) at (11.5,5) {$v_{n+1}$};   
     \node (An) at (11.5,3) {$a_{n+1}$}; % AJOUT DE JOHANNE
\end{scope}

    \node (3p1) at (9.2,5) {$\cdots$};
   % \node (3p1) at (5,1.5) {$\cdots$};
    \node (fake1) at (6.75,5.5) {$ $};
    \node (fake2) at (6.75,4.5) {$ $};
    \node (fake3) at (9.25,5.5) {$ $};
    \node (fake4) at (9.25,4.25) {$ $};
    \node (fake5) at (2.75,2.8) {$ $};
    \node (fake6) at (3.75,2.6) {$ $};
    \node (fake7) at (4.75,2.4) {$ $};

%% Forced edges
\begin{scope}[>={Stealth[black]},
              every node/.style={circle},
              every edge/.style={draw=black, dotted}]
    \draw [-,ultra thick,double] (S) to node[above,pos=0.50]{0}  (V0);
    \draw [-,ultra thick,double] (V0) to node[above,pos=0.50]{1}  (X1);
   
    \draw [-,ultra thick,double] (V0)  to node[above,pos=0.50]{1} (X1B);
    \draw [-,ultra thick,double] (X1) to node[above,pos=0.50]{2}   (V1);
    \draw [-,ultra thick,double] (X1B) to node[above,pos=0.50]{2}  (V1);
   \draw [-,ultra thick,double] (X2) to node[above,pos=0.50]{7}  (V2);
   \draw [-,ultra thick,double] (X2B)  to node[above,pos=0.50]{7}   (V2);
    \draw [-,ultra thick,double] (Xn) to  node[above,sloped,pos=0.50]{\tiny $7\frac{n}2$} (Vn);
    \draw [-,ultra thick,double] (XnB) to   node[above,sloped,pos=0.50]{\tiny $7\frac{n}2$}  (Vn);
 
    \draw [-,ultra thick,double] (X2) to node[above,pos=0.50]{4} (B2);    % AJOUT DE JOHANNE
    \draw [-,ultra thick,double] (B2) to node[left,pos=0.250]{5}  (V1);    % AJOUT DE JOHANNE
    \draw [-,ultra thick,double] (V0)  to node[left,pos=0.50]{1}  (A0); % AJOUT DE JOHANNE
    \draw [-,ultra thick,double] (V1)  to node[left,pos=0.50]{3}  (A1); % AJOUT DE JOHANNE
    \draw [-,ultra thick,double] (V2)  to node[left,pos=0.50]{8}   (A2); % AJOUT DE JOHANNE
    \draw [-,ultra thick,double] (Vn) to  (An); % AJOUT DE JOHANNE

\end{scope}

%% k copies
\begin{scope}
%    \draw [-,ultra thick] (w1) to  (w3);
 %   \draw [-,ultra thick] (w3) to  (t2);
    \draw [-,  thick,dashed] (t1) to node[above,sloped,pos=0.50]{\tiny $(2m,L+5)$}   (z1);
    \draw [-,  thick,dashed] (z2) to node[above,sloped,pos=0.40]{\tiny $(2m,L+5)$}  (t1);
     \draw [-,  thick,dashed] (zn) to node[above,sloped,pos=0.60]{\tiny $(2m,L+5)$}  (t1);
  
     \draw [-, thick,dashed] (A0) to node[left,sloped,above,pos=0.50]{\tiny $(k,L+5)$}    (t4);
     \draw [-, thick,dashed] (A1) to node[left,sloped,above,pos=0.50]{\tiny $(k,L+5)$ }  (t4);
     \draw [-, thick,dashed] (A2) to node[left,sloped,above,pos=0.50]{\tiny $(k-1,L+5)$ }   (t4);
     \draw [-,thick,dashed] (An) to node[left,sloped,above,pos=0.50]{\tiny $(k-\frac{n}{2},L+5)$ }   (t4);

\end{scope}

%% 2 copies
\begin{scope}[>={Stealth[black]},
              every node/.style={circle},
              every edge/.style={draw=blue, thick}]
    \draw [-,thick,dashed] (B2)  to node[left,pos=0.50]{\tiny $(2,5)$}  (t1);    % AJOUT DE JOHANNE

\end{scope}

%% normal
 \begin{scope}
     \draw [-] (X1) to   node[right,pos=0.80]{\tiny $L+4$} (z1);
     \draw [-] (X1B) to node[left,pos=0.80]{\tiny $L+4$} (z1);
     \draw [-] (X2) to node[left,pos=0.50]{\tiny $L+4$}  (z2);
     \draw [-] (X2B) to node[right,pos=0.80]{\tiny $L+4$}  (z2);
     \draw [-] (Xn) to  node[right,pos=0.80]{\tiny $L+4$}   (zn);
     \draw [-] (XnB) to node[left,pos=0.80]{\tiny $L+4$}  (zn);
       \draw [-] (V1)  to node[above,pos=0.50]{  3}   (X2);
       \draw [-] (V1) to node[above,pos=0.50]{  6}  (X2B);
       \draw [-] (fake3) to  (Xn);% AJOUT DE JOHANNE
     \draw [-] (fake4) to  (XnB);% AJOUT DE JOHANNE
 \end{scope}

\end{tikzpicture}
\caption{A piece of graph $G$ corresponding to the $n$ variables. 
Double thick edges denote ($k+1$-parallel) edges (forced edges). The notation $(\alpha,\tau)$ for the dashed edges  indicates that there are $\alpha$ parallel edges appearing at time $\tau$. For  readability there are several copies of $t$. As in the proof, $L=7\frac{n}2$. }
  \label{fig:PSPACE:variable}
\end{center}
\end{figure}
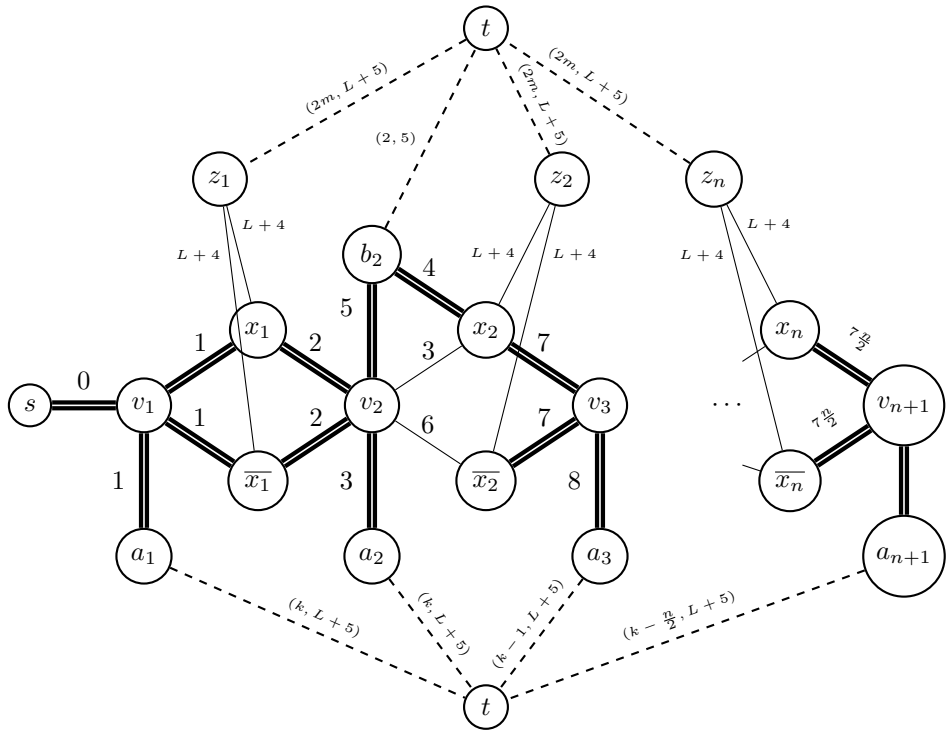

\begin{figure}[!h]
\begin{center}
\begin{tikzpicture}
\begin{scope}[every node/.style={circle,thick,draw}]
     \node (Vn) at (0,5) {$v_{n+1}$};
    \node (w1) at (2.5,5) {$w $} ;
    \node (c1) at (4.8,6.5) {$c_1$} ;
    \node (c2) at (4.8,5) {$c_2$} ;
    \node (cm) at (4.8,3) {$c_m$} ;
     \node (t1) at (2.5,7.5) {$t$} ;

     \node (t2) at (11.5,5.5) {$t$} ;
          \node (X1) at (6.5,3.5) {$x_{i}$} ;
     \node (X2) at (6.8,5.5) {$\overline{x_{j}}$} ;
     \node (X3) at (6.5,7.5) {$\overline{x_{\ell}}$} ;
     \node (Z1) at (9,3.5) {$z_{i}$} ;
     \node (Z2) at (9,5.5) {$z_{j}$} ;
     \node (Z3) at (9,7.5) {$z_{\ell}$} ;

\end{scope}

    \node (3p1) at (5,4) {$\vdots$};

%% Forced edges
\begin{scope}
    \draw [-,ultra thick,double] (Vn) to node[above,pos=0.50]{\tiny $L+1$}  (w1);
    \draw [-,ultra thick,double] (c2) to  node[sloped,above,pos=0.50]{\tiny $L+3$} (X1);
    \draw [-,ultra thick,double] (c2) to  node[sloped,above,pos=0.50]{\tiny $L+3$}(X2);
    \draw [-,ultra thick,double] (c2) to  node[sloped,above,pos=0.50]{\tiny $L+3$} (X3);
\end{scope}

%% k copies
\begin{scope}
    \draw [-] (w1) to  node[left,pos=0.50]{\tiny $L+2$}  (t1);
    \draw [-] (X1) to  node[above,pos=0.50]{\tiny $L+4$}  (Z1);
    \draw [-] (X2) to  node[above,pos=0.50]{\tiny $L+4$}  (Z2);
    \draw [-] (X3) to  node[above,pos=0.50]{\tiny $L+4$}  (Z3);

 \end{scope}

%% 2 copies
\begin{scope}
    \draw [-,thick,dashed] (w1) to  node[sloped,above,pos=0.50]{\tiny  $(2,L+2)$} (c1);
    \draw [-,thick,dashed] (w1) to  node[sloped,above,pos=0.60]{\tiny $(2,L+2)$} (c2);
    \draw [-,thick,dashed] (w1) to  node[sloped,above,midway]{\tiny $(2,L+2)$} (cm);
\end{scope}

%% 2 copies
\begin{scope}
    \draw [-,thick,dashed] (Z1) to  node[sloped,dashed,above,pos=0.50]{\tiny  $(2m,L+5)$} (t2);
    \draw [-,thick,dashed] (Z2) to  node[sloped,dashed,above,pos=0.50]{\tiny $(2m,L+5)$} (t2);
    \draw [-,thick,dashed] (Z3) to  node[sloped,dashed,above,midway]{\tiny $(2m,L+5)$}(t2);
\end{scope}

%% normal

\end{tikzpicture}
\caption{A piece of graph $G$ corresponding to the clauses. Double thick edges denote ($k+1$-parallel) edges (forced edges). The notation $(\alpha,\tau)$ for the dashed edges  indicates that there are $\alpha$ parallel edges appearing at time $\tau$. For readability, there are several copies of the target vertex $t$. As in the proof, $L=7\frac{n}2$. We assume that clause $c_{2}=(x_{i}\lor \overline{x_{j}}\lor \overline{x_{\ell}})$.}
  \label{fig:PSPACE:clause}
\end{center}
\end{figure}
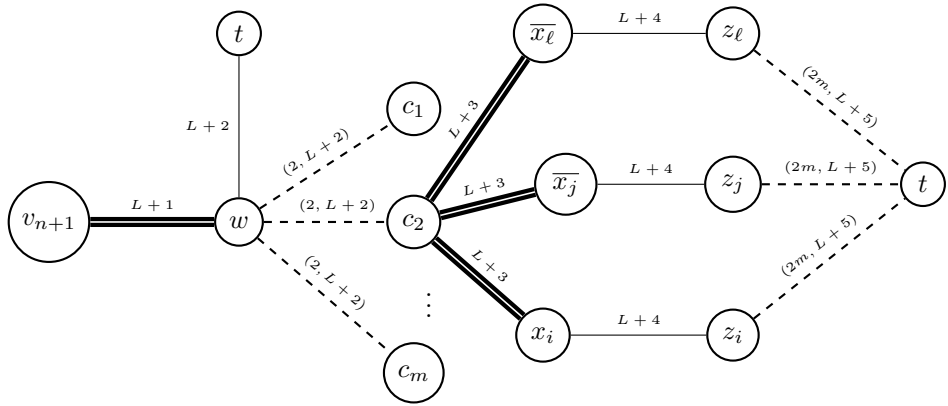

\begin{itemize}

 \item for each variable $x_{i}$ where $i$ is odd, we introduce an "existential" {gadget}. We put a diamond on 4 vertices $v_{i},x_i,\overline{x_i},v_{i+1}$, with 4 {\it forced} edges $(v_{i},x_i,\impairUN)$, $(v_{i},\overline{x_i},\impairUN)$, $(x_i,v_{i+1},\impairDEUX)$ and $(\overline{x_i},v_{i+1},\impairDEUX)$. 
 Moreover, we add one forced edge $(v_{i},a_i,\impairUN)$,  and {$\K-\lfloor \frac{i}2\rfloor$} parallel edges $(a_i,t,\FIN)$. % for each  $i$  $2\leq i\leq n+1$,

\item for each variable $x_{i}$ where $i$ is even, we introduce a "universal" {gadget}. We put a diamond on 4  vertices  $v_{i},x_i,\overline{x_i},v_{i+1}$ with one edge $(v_{i},x_i,\pairUN)$, one edge  $(v_{i},\overline{x_i},\pairQUATRE)$,  2 {\it forced} edges  $(x_i,v_{i+1},\pairCINQ)$ and $(\overline{x_i},v_{i+1},\pairCINQ)$.  Moreover, there  are 2 {\it forced} edges $(x_i,b_{i}, \pairDEUX)$ and $(b_{i}, v_{i}, \pairTROIS)$.   Vertex $b_{i}$ is linked to $t$ with $2$  parallel edges $(b_{i},t,\pairDEUX)$.  Finaly,  we add one forced edge $(v_{i},a_i,\pairUN)$,  and {$\K-\lfloor \frac{i}2\rfloor$}  parallel edges $(a_i,t,\FIN)$. % for each  $i$  $2\leq i\leq n+1$,
\item  for each  $i$,   $1\leq i\leq n$,  we have one  edge $(\overline{x_i},z_i,\FINmoinsUN)$,  one  edge $(x_i,z_i,\FINmoinsUN)$, and %\textcolor{blue}{$\K-\lfloor \frac{i}2\rfloor$} 
{$2m$} parallel edges $(z_i,t,\FIN)$.

%one  edge $(\overline{x_i},z_i,\FINmoinsUN)$,  and  $\KK$ parallel edges $(z_i,t,\FIN)$
%\item one  edge $({x_i},{z_i,\FINmoinsUN)$,  and  $\KK$ parallel edges $(z_i,t,\FIN)$ for each  $i$,   $1\leq i\leq n$.
\end{itemize}
%\textcolor{blue}{I think $\K-\lfloor \frac{i}2\rfloor$  parallel edges $(a_i,t,\FIN)$ for both cases $i$ odd or even; and $\K-\lfloor \frac{i}2\rfloor$ parallel edges $(z_i,t,\FIN)$; $\K = 2m+\frac{n}{2}$.} 
%\textcolor{blue}{Johanne:I do not believe} 
We introduce the edges associated with the clauses (see Figure~\ref{fig:PSPACE:clause}). For each clause $C_{i}=(\ell_{1}\lor\ell_{2}\lor \ell_{3})$, $1\leq i \leq m$, we have
\begin{itemize}
 \item  one forced edge  $(v_{n+1},w,\FINmoinsQUATRE)$,  one edge $(w,t,\FINmoinsTROIS)$
 
 \item two parallel edges $(w,C_{i},\FINmoinsTROIS)$; 
 \item a forced edge $(C_{i},\ell,\FINmoinsDEUX)$, for each $\ell \in \{\ell_{1},\ell_{2},\ell_{3}\}$.
\end{itemize}

  Note that $G$ is constructed in polynomial time.\\  %We prove that Satisfier wins in φ if and only if Traveller wins. 

    The remaining part of the proof is given in Appendix \ref{app:pspace}. We now give an intuition on how the reduction works. Intuitively, if the formula is true, Traveller follows the diamonds (vertex-gadgets) to reach $v_{n+1}$. When Traveller is at $v_i$ for some even integer $i$, Blocker must block exactly one edge between $(v_i,x_i,\pairUN))$ or $(v_{i},\overline{x_i},\pairQUATRE)$ if they do not want to loose. They choose the one they want corresponding to the universal variable $\forall x_i$. 
    For an odd $i$, Traveller chooses to go to $x_i$ or $\overline{x_i}$ according to the assignment they want to give to the existential $\exists x_i$.

    When Traveller is at node $x_i$ or $\overline{x_i}$, Blocker must not block the edge leading to $z_i$. Otherwise, they lose. Therefore, at $v_{n+1}$, Traveller has seen exactly $\frac{n}{2}$ blocked edges and has visited either $x_i$ or $\overline{x_i}$ for all $i$ (if Blocker had blocked more edges, Traveller would go to $a_{n+1}$ and win).

    Blocker then must block edge $(w,t,\FINmoinsTROIS)$, and cannot prevent Traveller from reaching a clause vertex $C_j$ for some $j$. From $C_j$, Traveller can visit a vertex $x_i$ or $\overline{x_i}$ that they already visited before (because the formula is satisfiable and they chose the assignments they could depending on the choices of Blocker accordingly). Consequently,  from visited vertex $x_i$ or $\overline{x_i}$, Traveller can go to vertex $z_i$ and take one of the parallel edges $(z_i,t,\FIN)$ to reach $t$ and win.\\

    If the formula is not true, thanks to the path to $v_{n+1}$, Blocker can force the assignment for universal variables {in such a way that,  regardless of Traveller's decisions on existential variables, the formula remains unsatisfied.}
  %  whatever choice Traveller makes on an existential variable, the formula at the end is not satisfied. At $v_{n+1}$, Traveller has seen exactly $\frac{n}{2}$ blocked edges. 
  Blocker {blocks} edge $(w,t)$. Blocker then can force Traveller to visit a vertex $C_j$ which is not satisfiable by the truth assignment, using $2(m-1)$ edges to block access to the other clauses vertices. Traveller then must visit a vertex $x_i$ or $\overline{x_i}$ that they have not visited before. This allows Blocker to block the edge from that node leading to $z_i$, preventing Traveller from reaching $t$.

%The proof of this Lemma is in Appendix~\ref{app:pspace}.
 
\end{proof}

 In Appendix~\ref{app:k=1}, we provide a polynomial algorithm for the case $k=1$.

\section{NP-hardness for a constant number of blocked edges for static CTP and Li-TCTP}\label{sec:static}

In this section, we prove that the Canadian Traveller Problem is NP-complete in the static setting even when $k=4$, using a reduction from 3-SAT. We then show how this reduction can be adapted to demonstrate that Li-TCTP is hard for $k=2$.

\begin{theorem}\label{th:static}
$k$-CTP is NP-hard even for $k=4$.
\end{theorem}
\begin{proof}
    Let us consider an instance of 3-SAT on $n$ binary variables $x_1,\dots,x_n$ and $m$ clauses $C_1,\dots,C_m$. 

    We build a graph -  a multigraph (we could remove multi-edges, but we keep them for the ease of explanation) - such that Traveller can guarantee to reach $t$ by some given time $T$ (when up to $4$ edges can be blocked) if and only if the formula is satisfiable.
     The intuition of using a constant number of edges is that the Traveller's path will walk through each variable (as in our previous construction) to choose their truth assignment. Then, they will walk through each clause once after another, choosing a {literal} their assignment satisfies. Suppose they visit an unsatisfied {literal}. In that case, Blocker will block $3$ edges to force Traveller to return to the corresponding variable assignment and then block the edge leading to the target.

   The graph contains two vertices, $s$ and $t$, and we consider a large integer $M$ (larger than $2n+2m$). {As previously, we use the term `forced edge'} to denote $5=k+1$ copies of an edge. A forced edge then cannot be blocked. The construction is illustrated in Figure~\ref{fig:reductionStatic}.
    To each variable $x_i$, the following gadget is associated: 
    \begin{itemize}
        \item 6 vertices $v_{i-1},v_i,x_i,\overline{x_i},z_i,\overline{z_i}$
        \item 4 forced edges of length 1: $\EDGE{v_{i-1}}{x_i}$, $\EDGE{v_{i-1}}{\overline{x_i}}$, $\EDGE{x_i}{v_{i}}$, $\EDGE{\overline{x_i}}{v_{i}}$
        \item 2 edges $\EDGE{x_i}{z_i}$ and $\EDGE{\overline{x_i}}{\overline{z_i}}$ of length $10M$
        \item Each $z_i$ and each $\overline{z_i}$ is linked to $t$ by two copies of an edge of length 0.
    \end{itemize}
Gadgets are linked together, meaning that $v_i$ is part of the gadgets of $x_i$  and $x_{i+1}$.

 Finally: 
\begin{itemize}
    \item $v_n$ is linked to a vertex $w$ by a forced edge of length $2m+27M$, and $w$ is linked to $t$ by 4 copies of an edge of length 0;
    \item we consider as starting vertex $s=v_0$.
\end{itemize}

Then, to each clause $C_j$ containing 3 literals $\ell_j^1$, $\ell_j^2$, $\ell_j^3$, is associated the following gadget:
\begin{itemize}
    \item 7 vertices $c_j$, $\alpha_j^1$, $\alpha_j^2$, $\alpha_j^3$, $\beta_j^1$, $\beta_j^2$, $\beta_j^3$ (note that $\beta_j^1$ and $\beta_j^2$ are not drawn in Figure~\ref{fig:reductionStatic} for the sake of readability);
    \item $c_j$ is linked to each $\alpha_j^s$ by a forced edge of length 1;
    \item each $\alpha_j^s$ is linked to $c_{j+1}$ (vertex of the gadget of the next clause) by three copies of an edge of length 1;
    \item $\alpha_j^s$ is linked to $\beta_j^s$ by two copies of an edge of length $d_j=2m-2j+1$.
\end{itemize}
Finally, vertex $c_{m+1}$ (to which are linked the $\alpha$-vertices of the last clause) is linked to $t$ by a forced edge of length $18M$.

Then, the vertex-gadgets and the clause-gadgets are linked as follows:
\begin{itemize}
    \item $v_n$ is linked to $c_1$ by an edge of length $9M$.
    \item Vertex $\beta_j^s$ corresponds to a literal of clause $C_j$: we link $\beta_j^s$ to the corresponding literal $x_i$ or $\overline{x_i}$ in the vertex-gadgets by a forced edge of length $8M$. 
\end{itemize}
To conclude the construction of the 4-CTP instance, we set $T=27M+2n+2m$ and $k=4$.  

 The remaining part of the proof is given in Appendix \ref{app:statichardness}; we now give a very rough idea of how the reduction works. If the formula is satisfiable, Traveller follows the diamonds (vertex-gadgets) to reach $v_{n}$, following a truth value satisfying the formula. If Blocker blocks an edge, Traveller can reach $t$ on time (through $w$). Otherwise, Traveller traverses the clause gadgets. If they are not blocked, they reach $t$ after $c_{m+1}$ on time. Otherwise, Blocker has blocked the three copies of an edge between some $\alpha_j^s$ and $c_{j+1}$. Then Traveller can go to $\beta_j^s$, then to the corresponding literal $x_i$ or $\overline{x_i}$, visited before, so the edge towards $z_i$ or $\overline{z_i}$ is not blocked. Then Traveller reaches $t$ (on time), as only one more edge can be blocked.

  If the formula is not satisfiable, then it can be shown that Traveller has to follow a truth value in the variable gadgets and then go across the clause gadgets. Blocker does not block anything until Traveller goes to an $\alpha_j^s$ corresponding to a false literal. Then Blocker blocks the three copies of  $\EDGE{\alpha_j^s}{c_{j+1}}$. Traveller will have no choice but to go to the corresponding vertex $x_i$ or $\overline{x_i}$, unvisited before, and Blocker block the edge towards $z_i$ or $\overline{z_i}$, making Traveller too late to reach $t$ on time. 
\end{proof}

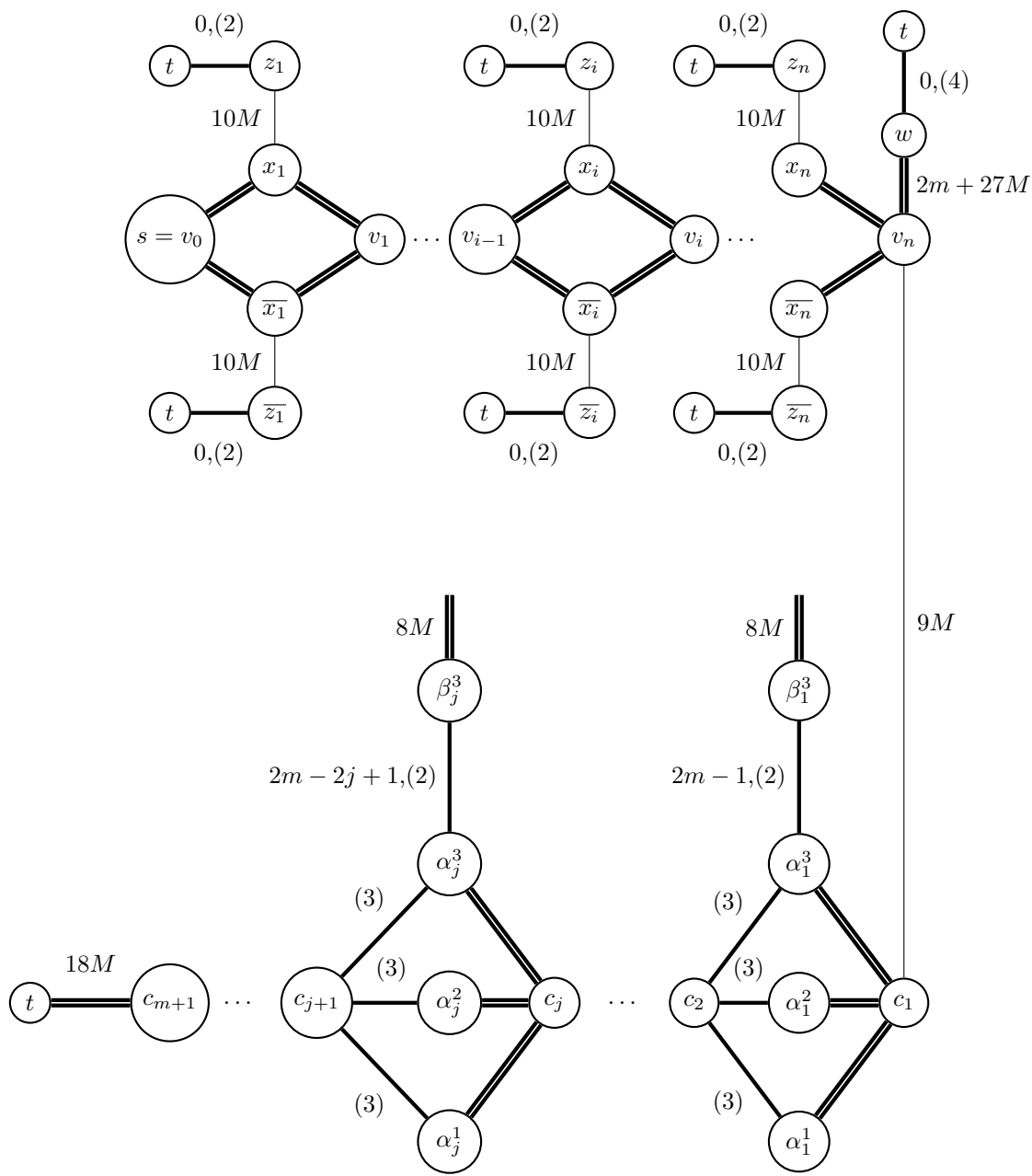
\begin{figure}[H]
\begin{center}
\begin{tikzpicture}
\begin{scope}[every node/.style={circle,thick,draw}]
%    \node (S) at (0,5) {$s$};
    \node (V0) at (1.5,5) {$s=v_0$};
    \node (X1) at (3,6) {$x_1$};
    \node (X1B) at (3,4) {$\overline{x_1}$};
    \node (V1) at (4.5,5) {$v_1$};
    \node (VI-1) at (6,5) {$v_{i-1}$};
    \node (XI) at (7.5,6) {$x_i$};
    \node (XIB) at (7.5,4) {$\overline{x_i}$};
    \node (VI) at (9,5) {$v_i$};
    \node (XN) at (10.5,6) {$x_n$};
    \node (XNB) at (10.5,4) {$\overline{x_n}$};
    \node (VN) at (12,5) {$v_n$};
    \node (W) at (12,6.5) {$w$} ;
    \node (TW) at (12,8) {$t$} ;
    \node (C1) at (12,-6) {$c_1$} ;
    \node (A1) at (10.5,-8) {$\alpha^1_1$} ;
    \node (A2) at (10.5,-6) {$\alpha^2_1$} ;
    \node (A3) at (10.5,-4) {$\alpha^3_1$} ;
    \node (B3) at (10.5,-1.5) {$\beta^3_1$} ;    
    \node (C2) at (9,-6) {$c_2$} ;
    \node (CJ) at (7,-6) {$c_{j}$} ;
    \node (A4) at (5.5,-8) {$\alpha^1_j$} ;
    \node (A5) at (5.5,-6) {$\alpha^2_j$} ;
    \node (A6) at (5.5,-4) {$\alpha^3_j$} ;
    \node (B6) at (5.5,-1.5) {$\beta^3_j$} ;
    \node (CJ+1) at (3.6,-6) {$c_{j+1}$} ;
    \node (CM+1) at (1.5,-6) {$c_{m+1}$} ;
    \node (T1) at (1.5,7.5) {$t$} ;
    \node (T2) at (6,7.5) {$t$} ;
    \node (T3) at (9,7.5) {$t$} ;
    \node (Z1) at (3,7.5) {$z_{1}$} ;
    \node (ZI) at (7.5,7.5) {$z_{i}$} ;
    \node (ZN) at (10.5,7.5) {$z_{n}$} ;
    \node (T4) at (1.5,2.5) {$t$} ;
    \node (T5) at (6,2.5) {$t$} ;
    \node (T6) at (9,2.5) {$t$} ;
    \node (T7) at (-0.5,-6) {$t$} ;
    \node (K1) at (3,2.5) {$\overline{z_1}$} ;
    \node (KI) at (7.5,2.5) {$\overline{z_i}$} ;
    \node (KN) at (10.5,2.5) {$\overline{z_n}$} ;
\end{scope}
%\draw [-] (V0) to node[right,pos=0.50]{2} (V2);
    \node (fake6) at (5.5,0) {$ $};
    \node (fake7) at (10.5,0) {$ $};
    \node (Dots) at (5.2,5) {$\dots$};
    \node (DotsB) at (9.7,5) {$\dots$};
    \node (DotsC) at (8,-6) {$\dots$};
    \node (DotsD) at (2.5,-6) {$\dots$};

%% Forced edges
\begin{scope}[>={Stealth[black]},
              every node/.style={circle},
              every edge/.style={draw=black, dotted}]
%    \draw [-,ultra thick,double] (S) to  (V0);
    \draw [-,ultra thick,double] (V0) to  (X1);
    \draw [-,ultra thick,double] (V0) to  (X1B);
    \draw [-,ultra thick,double] (X1) to  (V1);
    \draw [-,ultra thick,double] (X1B) to  (V1);
    \draw [-,ultra thick,double] (VI-1) to  (XI);
    \draw [-,ultra thick,double] (VI-1) to  (XIB);
    \draw [-,ultra thick,double] (XI) to  (VI);
    \draw [-,ultra thick,double] (XIB) to  (VI);
    \draw [-,ultra thick,double] (XN) to  (VN);
    \draw [-,ultra thick,double] (XNB) to  (VN);
    \draw [-,ultra thick,double] (VN) to node[right,pos=0.50]{$2m+27M$} (W);
    \draw [-,ultra thick,double] (A1) to  (C1);
    \draw [-,ultra thick,double] (A2) to  (C1);
    \draw [-,ultra thick,double] (A3) to  (C1);
    \draw [-,ultra thick,double] (A4) to  (CJ);
    \draw [-,ultra thick,double] (A5) to  (CJ);
    \draw [-,ultra thick,double] (A6) to  (CJ);
    \draw [-,ultra thick,double] (CM+1) to node[above,pos=0.50]{$18M$} (T7);
    \draw [-,ultra thick,double] (B3) to node[left,pos=0.50]{$8M$} (fake7);
    \draw [-,ultra thick,double] (B6) to node[left,pos=0.50]{$8M$} (fake6);
%    \draw [-,ultra thick,double] (VN) to  node[right,pos=0.50]{$\Delta$} (B00);
%    \draw [-,ultra thick,double] (Blogmj) to node[left,pos=0.30]{$\Delta$}  (fake6);
%    \draw [-,ultra thick,double] (Blogmj) to node[right,pos=0.50]{$\Delta$} (fake7);
%    \draw [-,ultra thick,double] (Blogmj) to  node[right,pos=0.30]{$\Delta$} (fake8);
\end{scope}

%% k copies
\begin{scope}[>={Stealth[black]},
              every node/.style={circle},
              every edge/.style={draw=black, thick}]
    \draw [-,ultra thick] (Z1) to node[above,pos=0.50]{0,(2)} (T1);
    \draw [-,ultra thick] (ZI) to  
    node[above,pos=0.50]{0,(2)} (T2);
    \draw [-,ultra thick] (ZN) to node[above,pos=0.50]{0,(2)} (T3);
    \draw [-,ultra thick] (K1) to node[below,pos=0.50]{0,(2)} (T4);
    \draw [-,ultra thick] (KI) to  
    node[below,pos=0.50]{0,(2)} (T5);
    \draw [-,ultra thick] (KN) to node[below,pos=0.50]{0,(2)} (T6);
    \draw [-,ultra thick] (W) to node[right,pos=0.50]{0,($4$)} (TW);
    \draw [-,ultra thick] (A1) to node[below left,pos=0.40]{(3)} (C2);
    \draw [-,ultra thick] (A2) to node[above,pos=0.40]{(3)} (C2);
    \draw [-,ultra thick] (A3) to node[above left,pos=0.40]{(3)} (C2);
    \draw [-,ultra thick] (A4) to node[below left,pos=0.40]{(3)} (CJ+1);
    \draw [-,ultra thick] (A5) to node[above,pos=0.40]{(3)} (CJ+1);
    \draw [-,ultra thick] (A6) to node[above left,pos=0.40]{(3)} (CJ+1);
    \draw [-,ultra thick] (A3) to node[left,pos=0.50]{$2m-1$,(2)} (B3);
    \draw [-,ultra thick] (A6) to node[left,pos=0.50]{$2m-2j+1$,(2)} (B6);
\end{scope}

%% two copies
\begin{scope}[>={Stealth[black]},
              every node/.style={circle},
              every edge/.style={draw=blue, thick}]
%    \draw [-,double] (w2) to  (c1);
%    \draw [-,double] (w2) to  (c2);
%    \draw [-,double] (w2) to  (cm);
\end{scope}

%% normal
\begin{scope}[>={Stealth[black]},
              every node/.style={circle},
              every edge/.style={draw=blue, thick}]
    \draw [-] (X1) to node[left,pos=0.50]{$10M$} (Z1);
    \draw [-] (XI) to node[left,pos=0.50]{$10M$} (ZI);
    \draw [-] (XN) to node[left,pos=0.50]{$10M$} (ZN);
    \draw [-] (X1B) to node[left,pos=0.50]{$10M$} (K1);
    \draw [-] (XIB) to node[left,pos=0.50]{$10M$} (KI);
    \draw [-] (XNB) to node[left,pos=0.50]{$10M$} (KN);
    \draw [-] (VN) to node[right,pos=0.50]{$9M$} (C1);
%    \draw [-] (X1) to  (z);
\end{scope}

\node (fake1) at (9,-2) {$ $};
\node (fake2) at (5,-2) {$ $};
% \draw[dotted,thick] (B00) to (fake1);
% \draw[dotted,thick] (B00) to (fake2);
% \draw[dotted,thick] (fake2) to (Blogmj);
% \draw[dotted,thick] (fake1) to (Blogmj);

\end{tikzpicture}
\caption{Construction of the graph. Double-thick edges represent forced edges. For (single) thick edges, the number in parenthesis represents the number of parallel edges. The number (without parenthesis) represents an edge length; it is 1 when no number exists. For the sake of readability, vertices $\beta_j^1$ and $\beta_j^2$ are not depicted.}
  \label{fig:reductionStatic}
\end{center}
\end{figure}

The proof of this result can be adapted in the Li-TCTP setting, using $k=2$ blocked edges (see Appendix~\ref{app:utctp2}).

\begin{theorem}\label{th:utctp2}
    Li-TCTP is NP-hard, even if $k=2$. % is order \textcolor{blue}{To complete}
\end{theorem}

%\subsubsection*{Acknowledgement}

\newpage

\bibliography{ref}

\appendix

\section{Discussion on our model and result (for static graphs) and the ones in \cite{bar1991canadian}.}
\label{app:difference}

In this work, we focus on the CTP problem where the status of an edge (being blocked or not) does not change while Traveller explores the graph (see Definition~\ref{def:staticctp} for static graphs).
{Initially, Traveller is unaware of whether an edge is blocked. However, upon reaching an endpoint of an edge, they learn its status, which remains unchanged thereafter.}
%whether an edge is blocked or not is initially unknown to Traveller, but when Traveller reaches an endpoint of an edge they know if it is blocked or not, and this does not change afterwards.
In this model, which is the one studied in~\cite{papadimitriou1991shortest} where the problem is introduced, we prove that CTP is $NP$-hard for $k=4$ (see Theorem~\ref{th:static}). To our knowledge, this is the first hardness proof for some constant $k$ for this problem.

In contrast, in~\cite{bar1991canadian}, an algorithm solving CTP in polynomial time for fixed $k$ is provided. This algorithm solves the problem in a slightly different model where the status of an edge can change from unblocked to blocked. % (i.e. once an edge is blocked, it is forever, but the other way around is no longer the case).
Notably, the recursion of the algorithm does not use the knowledge of the previously visited vertices and, hence, what non-blocked edges were seen.

Intuitively, {the complexity differences between these models can be intuitively understood as follows:} 
%to give a rough idea of why the complexity is different in the two models,
in the latter model {it suffices for Traveller to remember only the set of} blocked edges encountered so far. % \textcolor{blue}{because} a non-blocked edge might be blocked later, seeing a non-blocked edge does not give any information for the future, so it is pointless to remember them.
{In deed, if Traveller encounters a non-blocked edge, it might become blocked later. Observing a non-blocked edge at any point does not provide predictive information for subsequent edges, rendering it unnecessary to retain such information.}
{Consequently, the number of possible configurations remains polynomial for fixed $k$.}  %leading to a polynomial number of possible configurations for fixed $k$. 
{Conversely}, in the model of~\cite{papadimitriou1991shortest}, Traveller shall also remember the non-blocked edges seen so far in their exploration, as this is information they have learned, leading to an exponential number of configurations even for fixed $k$. This can be seen in the proof of Theorem~\ref{th:static}, where when reaching $v_n$ Traveller has seen non-blocked edges corresponding to a truth value (exponential number of different configurations, even without having seen any blocked edge yet).

%\textcolor{blue}{Mikaël:Interestingly, for the PSPACE hardness proof in~\cite{bar1991canadian}, they explain that it can be done by adapting the PSPACE hardness proof of~\cite{papadimitriou1991shortest} (without any description of this adaptation). The proof in~\cite{papadimitriou1991shortest} heavily relies on the fact that once an edge is seen, its weight cannot be changed. One can wonder if the problem where an edge can become blocked later is still PSPACE-hard.}

\section{Proof of Lemma~\ref{lem:acyclic2}}\label{app:lemacyclic}

{\bf Lemma~\ref{lem:acyclic2}.} {\it 
    For the Canadian Traveller Decision Problem in a weighted {DAG} $G$,   Traveller has a winning strategy to reach $t$ with arrival time at most $T$ if and only if $\pi_k(s) \leq T$. %\textcolor{blue}{Comment: I removed the sentence on Blocker, which was useless as we use if and only if}%Blocker has a winning strategy, i.e., a strategy to prevent Traveller from reaching $t$ within budget $T$, if $\pi_k(s) > T$.
    }

\begin{proof}
    Let us describe a strategy for  Traveller when $\pi_k(s) \leq T$: Traveller remembers how many edges were already blocked. Assume that Traveller is at node $u$, and $m_2$ edges are missing from $u$, and there were already $m_1$ edges missing before. Then, Traveller travels an existing edge $uw$ to minimize $(\pi_{k-m_1-m_2}(w)+d(u,w))$.

    We prove that for all $i = 1,\ldots, k$, for all $v$, Traveller has a strategy to go from $v$ to $t$ within budget $\pi_i(v)$ if Blocker can block at most $i$ arcs between $v$ and $t$. The proof is obtained by induction on the topological order of the vertices. {If $v$ has only $t$ as outneihbor}, the statement holds for $v$. {Otherwise, assume that the statement holds for all vertices greater than $v$ in the topological order}. 
    Assume that at node $v$, Blocker blocks $m$ arcs starting from $v$, then Traveller travels the arc $vu$ minimizing $d(v,u)+\pi_{i-m}(u)$. We have $$d(v,u)+\pi_{i-m}(u) \leq  \text{(m+1)-th min}_{u\in N^+(v)} \pi_{i-m}(u) +d(v,u)  \leq \pi_i(v)$$
    
    By the induction hypothesis, there is a strategy for Traveller to go from $u$ to $t$ within budget $\pi_{i-m}(u)$ when Blocker can block at most $i-m$ arcs between $u$ and $t$. It implies that Traveller has a strategy to go from $v$ to $t$ within budget $\pi_i(v)$.

    \vspace{0.3cm}
    If $\pi_k(s) >T$, there exists a winning strategy for Blocker: We prove that for all $i = 1,\ldots, k$, for all $v$, if Blocker can block $i$ arcs between $v$ and $t$, Blocker has a strategy to force Traveller to spend budget $\pi_i(v)$ to go from $v$ to $t$. The proof is obtained by induction on the topological order of vertices. {If $v$ has only $t$ as outneighbor}, Blocker blocks $i$ arcs from $v$ to $t$ with minimum weights, and the statement holds for $v$. {Otherwise, assume that the statement holds for all vertices greater than $v$ in the topological order}. 
    At node $v$, there exists $m\leq i$ such that $\text{(m+1)-th min}_{u\in V^+(v)} \pi_{i-m}(u) +d(v,u) = \pi_i(v)$. Blocker blocks $m$ arcs $vu$ that minimizes $\pi_{i-m}(u) +d(v,u)$. Assume that Traveller takes an existing arc $vw$, then 
    $$d(v,w)+\pi_{i-m}(w) \geq \text{(m+1)-th min}_{u\in V^+(v)} \pi_{i-m}(u) +d(v,u) = \pi_i(v)$$ 
    By the induction hypothesis, there exists a strategy for Blocker to force Traveller to spend $\pi_{i-m}(w)$ to go from $w$ to $t$ when Blocker still can block $(i-m)$ arcs between $w$ and $t$. Therefore, if Blocker can block $i$ arcs between $v$ and $t$, Blocker has a strategy to force Traveller to spend at least $\pi_i(v)$ to go from $v$ to $t$.

\end{proof}

\section{Remaining part of the proof of Theorem \ref{th:pspace}}\label{app:pspace}

   %Let us now show this. 
   We first consider the case where the formula is satisfiable and consider a truth assignment $\sigma$ satisfying it. Traveller starts from $s$, goes to $v_{1}$ at time $1$.  Traveller goes from $v_{1}$ to $v_{n+1}$, going through $x_i$ or $\overline{x_{i}}$.   We can divide this path into $n$ parts.

 At time $1$, Traveller goes through $x_{1}$ if $x_{1}$ is true in $\sigma$, through $\overline{x_{1}}$ otherwise. So, they reach $v_{2}$ at time $3$. Notably, the edges $(v_{1},x_1,1)$, $(v_{1},\overline{x_1},1)$, $({x_1},v_{2}, 2)$, $(\overline{x_1},v_{2}, 2)$ are forced edges, and Blocker cannot prevent Traveller from reaching vertex $v_{2}$.
 
Upon reaching vertex $v_{2}$ at time $3$,  Blocker can choose to block edges.
If Blocker blocks  edges $(v_{2},x_2,3)$ and $(v_{2},\overline{x_2},6)$, Traveller can cross the forced edge $(v_{2},a_2,3)$ and one of the $\KmoinsUN$ parallel edges $(a_{2},t,\FIN)$. % among the $\KmoinsUN$ available.
Since there are $\KmoinsUN$  parallel edges $(a_{2},t,\FIN)$, Blocker cannot block all of them simultaneously, as it can now only block $\KmoinsDEUX$ edges. So Traveller can reach $t$. So Blocker cannot block both edges $(v_{2},x_2,3)$ and $(v_{2},\overline{x_2},6)$ simultaneously.  
%Suppose Blocker decides to block just one of those edges; let us say edge $(v_{2},x_2,3)$ without losing generality. So,  Traveller goes to  $\overline{x_2}$, and they reach $v_{3}$ at time $8$.  Now,    
%they go to $a_{3}$ via forced edge $(v_{3},a_{3},8)$. Next, Traveller will take one of the $\KmoinsUN$ parallel edges, and Blocker cannot block them all because now it can only block $\KmoinsUN$ edges. Traveller will be able to reach $t$.
%So Blocker cannot block any of these edges $(v_{2},x_2,3)$ and $(v_{2},\overline{x_2},6)$.
%So at time $3$, Traveller is at vertex $v_{2}$, passes through vertex $x_{2}$ (at step $4$) and arrives at $b_{2}$ at step $5$. Blocker must block the two edges $(b_{2},t,5)$.
%Otherwise, Traveller has won. Next, Traveller passes through the vertices $v_{2}$, $\overline{x_{2}}$, $v_{3}$. Moreover, now Blocker can only delete $\KmoinsDEUX$.
Suppose Blocker does not block any of these edges $(v_{2},x_2,3)$ and $(v_{2},\overline{x_2},6)$, Traveller is at vertex $v_2$, passes through vertex  $x_{2}$ (at step $4$) and arrives at $b_{2}$ at step $5$. Blocker must block the two edges $(b_{2},t,5)$.
Otherwise, Traveller has won. Next, Traveller passes through the vertices $v_{2}$, $\overline{x_{2}}$. Blocker must block edge $(\overline{x_2},z_2,\FINmoinsUN)$. Otherwise, Traveller can travel the edge $(\overline{x_2},z_2,\FINmoinsUN)$ and one of the $\KmoinsUN$ parallel edges $(z_i,t,\FIN)$ to reach $t$. 
Traveller at vertex $\overline{x_{2}}$ travels edges $(\overline{x_{2}},v_3,7)$,  $(v_3,a_3,8)$ and one of $\KmoinsUN$ parallel edges $(a_3,t,\FIN)$ to reach $t$.
Therefore, Blocker must block exactly one of two edges $(v_{2},x_2,3)$ and $(v_{2},\overline{x_2},6)$, if they do not want to lose. The other edge will represent the assignment chosen for the universal quantifier of variable $x_2$. Traveller takes the non-blocked edge to reach $x_2$ or $\overline{x_2}$, then follows a forced edge $(x_2,v_3,7)$ or $(\overline{x_2},v_3,7)$ to reach $v_3$ at time $7$.

Now, we can generalise the previous reasoning. Traveller will give an assignment $\sigma$ to variables $x_1,\ldots,x_n$, from 1 to $n$, making each time sure that it can be completed (i.e. each time there is a choice to make for a $\exists x_i$, we know that previous choices ensure that a valid choice exists to finish the assignment). We will see by direct induction that, when {Traveller reaches $v_i$, they have} seen $\frac{i-1}{2}$ blocked edges. The assignment is done as follows:

 \begin{itemize}
 \item  If  $i$ is odd, we have assigned a value for each $x_j$ with $j\le i$. There exists an assignment for $x_i$ (either positive or negative), and Traveller  traverses from $v_{i}$ to $v_{i+1}$, passing through $x_{i}$ if $x_{i}$ is true in $\sigma$, through $\overline{x_{i}}$ otherwise. Thus, Traveller arrives in $v_{i+1}$ at time $\impairTROIS$, and Blocker does not block any edge. Moreover, the assignment can still be completed.
 \item If  $i$ is even, Traveller arrives at vertex $v_{i}$ at time $\pairUN$. With the same reasoning {as} for $i=2$, we know that if Blocker blocks either both or none of the two edges going to $x_i$ and $\overline{x_{i}}$, we have a direct winning strategy. %\footnote{Do we want to state the reasonning again? }
% It aims to reach $v_{i+1}$ while ensuring that Blocker must delete two edges.  With this aim, it traveller traverses  the $v_{i}-v_{i+1}$ path 
%$v_{i} -x_{i} - b_{i} - v_{i} - \overline{x_{i}}  -v_{i+1}$. 
Suppose Blocker deletes one of those edges, {Traveller goes} through the other one and assigns the corresponding assignment to $x_i$. Moreover, {Traveller has} seen one more blocked edge. Hence, {Traveller sees} a blocked edge for every two variables.

  \end{itemize}

By choosing assignments $\sigma$ on $x_i$ (for $i$ odd), adapted depending on the forced choices for the assignments $x_i$ (for $i$ even), we ensure that each clause has at least one true literal.

Traveller arrives at vertex $v_{n+1}$ at time $\FF$, and Blocker has blocked {$\frac{n}{2}$} edges. Subsequently, Traveller reaches vertex $w$ by traversing the forced edge $(v_{n+1},w,\FINmoinsQUATRE)$ and is there at time $\FINmoinsTROIS$. Now, {Blocker} can block at most {$2m$} more edges.

Blocker must block edge $(w,t,\FINmoinsTROIS)$ (otherwise {there is} a direct edge to $t$), and it can no longer prevent Traveller from accessing a vertex among ${C_{1}, \dots, C_{m}}$. Suppose Traveller reaches vertex $C_{j}=(x_{\alpha}\lor x_{\beta}\lor x_{\gamma})$. Since $\sigma$ satisfies all the clauses, one of the variables (by symmetry, say $x_{\alpha}$) evaluates to true by $\sigma$. Since Traveller has already visited vertex $x_{\alpha}$, this implies that the edge $(x_{\alpha},z_{\alpha},\FINmoinsUN)$ has not been blocked.
At this moment, Blocker can block at most $2m-1$ edges. Subsequently, all the parallel edges $(z_{\alpha},t,\FIN)$ cannot be blocked by Blocker anymore because {there are} $2m$ copies to $t$. Thus, Traveller arrives at vertex $t$.\\

 Now, let us consider the case when the formula is not satisfiable. We show how {Blocker} can trap Traveller. Traveller starts from $s$, goes to $v_1$ (at time $1$), then to $x_1$ or $\overline{x_1}$ (at time $2$). Blocker does not block the edge $(x_1,z_1,\FINmoinsUN)$. If Traveller takes this edge, {Blocker blocks} the {$\KmoinsUN$} copies of $(z_1,t,\FIN)$, and Traveller cannot reach $t$. So, Traveller is forced to continue to go to $v_2$ (it arrives at this vertex at time $3$).

As the formula is not satisfiable, whatever the assignment for $x_1$ is, there exists an assignment for $x_2$ that keeps the formula unsatisfiable.
When Traveller is at vertex $v_2$, Blocker blocks one of two edges $(v_2,x_2,3)$ or $(v_2,\overline{x_2},6)$ corresponding to the other assignment of $x_2$. That way, the only choice for Traveller is to go through the edge, making the formula unsatisfiable. 
\begin{itemize}
    \item If that way Traveller reaches $x_2$. Blocker does not block edge $(x_2,z_2,\FINmoinsUN)$. If Traveller goes though edge $(x_2,z_2,\FINmoinsUN)$, Blocker can block all $\KmoinsUN$ parallel edges $(z_2,t,\FIN)$ to trap Traveller. If Traveller goes through edge $(x_2,b_2,4)$, Blocker blocks the two edges going to $t$, and Traveller has no path to reach $t$ anymore.    
    Therefore, Traveller must travel the forced edge $(x_2,v_3,7)$ to reach $v_3$ at time $7$.
    \item If that way Traveller reaches $\overline{x_2}$. Blocker does not block edge $(\overline{x_2},z_2,\FINmoinsUN)$. If Traveller goes though edge $(x_2,z_2,\FINmoinsUN)$, Blocker can block all $\KmoinsUN$ parallel edges $(z_2,t,\FIN)$ to trap Traveller. Therefore, Traveller must travel the forced edge $(\overline{x_2},v_3,7)$ to reach $v_3$ at time $7$.
\end{itemize}

We continue using the same strategy {for Blocker} until Traveller reaches $v_{n+1}$: for each odd $i$, Traveller chooses an assignment, and for each even $i$, Blocker ensures that the choice they force keeps the formula unsatisfiable. When Traveller reaches $v_{n+1}$, {Blocker has blocked} exactly $\frac{n}{2}$ edges.
%In the worst-case scenario, we have blocked $\frac n2$ edges.  

Traveller will reach vertex $w$ using a forced edge $(v_{n+1},w,\FINmoinsQUATRE)$. Next, {Blocker needs} to block edge $(w,t,\FINmoinsTROIS)$; otherwise, Traveller can reach $t$.
Now {Blocker has} $2m-1$ edges that {they} can block.

As the formula is not satisfiable, let $C_j$ be a clause not satisfied by $\sigma$. {Blocker blocks} $2m-2$ edges: $(w,C_{j'},\FINmoinsTROIS)$ with $j' \neq j$. At this moment, {Blocker has} one last edge to block. Traveller is forced to go to $C_j$ and to choose one literal $\ell$ in $C_j$. As $\sigma$ does not satisfy $C_j$, Traveller has not visited $\ell$ when going through the diamonds. So {Blocker blocks} edge $(\ell,z_{\ell},\FINmoinsUN)$, and Traveller is blocked there and cannot reach~$t$.

This completes the proof.

\section{Case $k=1$ for (Li-TCTP)}\label{app:k=1}
\begin{theorem}\label{lem:U-TCTP}
    Locally-informed Temporal Canadian Traveller Problem is solvable in polynomial time when $k=1$.
\end{theorem}

\begin{proof}
Let us define, for a vertex $v\neq t$ and an incident edge $(v,u,\tau,d)$, $\mu(v,e)$ as the latest time Traveller can leave $v$ to reach $t$ at $T$ or before when the edge $e$ is blocked (and no other edge is blocked). We then define $\lambda_1(v)=\min\{\mu(v,e): e=(v,u,\tau,d)\in E\}$, which is the latest time Traveller can leave $v$ to be able to reach $t$ when an edge incident to $v$ is removed.

We then consider the following Dijkstra-like algorithm:
\begin{enumerate}
    \item $S\leftarrow \{t\}$
    \item $\pi_1(t)\leftarrow T$ (with $T=\infty$ if no time bound required to reach $t$)
    \item $\nu_1(v)=\max_{(v,t,\tau,d)\in E}\{\tau: \tau+d\leq \pi_1(t)\}$ for all $v\neq t$ ($-\infty$ if no edge satisfying the condition)
    \item $\pi_1(v)\leftarrow \min\{\lambda_1(v),\nu_1(v)\}$ for all $v\neq t$
    \item while $s\not\in S$:
    \begin{enumerate}
        \item $v^*\leftarrow argmax\{\pi_1(v):v\not\in S\}$
        \item $S\leftarrow S\cup\{v^*\}$
        \item $\nu_1(v)\leftarrow \max\{\nu_1(v),\max_{(v,v^*,\tau,d)\in E}\{\tau: \tau+d\leq \pi_1(v^*)\}\}$ for $v\not\in S$
        \item $\pi_1(v)\leftarrow \min\{\pi_1(v),\nu_1(v)\}$ for all $v\not\in S$.
    \end{enumerate}
\end{enumerate}
Note that once $v^*$ is added to $S$, then $\pi_1(v^*)$ does not change anymore. Note also that, by construction, for $v\not\in S$ we have at each iteration: $$\pi_1(v)= \min\{\lambda_1(v),\max_{u\in S, (v,u,\tau,d)\in E}\{\tau: \tau+d\leq \pi_1(u)\}\}$$
We show by induction that $\pi_1(v^*)$ is the latest time Traveller can leave $v^*$ to be able to reach $t$, i.e., they have a winning strategy if and only if they are at $v^*$ at $\pi_1(v^*)$ (or before).

The hypothesis is valid for $t$. Suppose that it is true for some $S$, and consider one iteration, when $v^*= argmax\{\pi_1(v):v\not\in S\}$ will be added to $S$.

First, assume that Traveller is at $v^*$ at $\pi_1(v^*)$. If an edge $e=(v^*,u,\tau,d)$  is blocked, then Traveller has a winning strategy since $\pi_1(v^*)\leq \lambda_1(v^*)\leq \mu(v^*,e)$. Otherwise, no edge incident to $v^*$ is blocked, and then $v^*$ uses an edge $(v,u,\tau,d)$ with $u\in S$, $\tau+d\leq \pi_1(u)$, and  $\pi_1(v^*)\leq \tau$. Then Traveller is at $u$ at time $\tau+d\leq \pi_1(u)$. By induction hypothesis, there is a winning strategy for them from $u$. 

Suppose that Traveller is at $v^*$ at $\tau^*>\pi_1(v^*)$. If $\pi_1(v^*)=\lambda_1(v^*)$, then there exists $e=(v,u,\tau,d)$ such that $\lambda_1(v^*)=\mu(v,e)=\pi_1(v^*)<\tau^*$. Blocker blocks this edge, and by definition of $\mu(v,e)$ Traveller cannot reach $t$ on time. Otherwise, Traveller follows a walk. While it stays at vertices $v$ outside $S$, Blocker acts similarly: if $\pi_1(v)=\lambda_1(v)$, Blocker blocks an edge incident to $v$ which prevents Traveller to reach $t$ on time: indeed, they are at $v$ not before $\tau^*$, and $\tau^*>\pi_1(v^*)\geq \pi_1(v)=\lambda_1(v)$ (so such an edge exists). 

Otherwise, at some point, Traveller will go inside $S$. Let us consider the first time they enter $S$, using an edge $(v,u,\tau,d)$ from $v\not\in S$ to $u\in S$. Then, we know that $\pi_1(v)\neq \lambda_1(v)$, so $\pi_1(v)= \max_{u'\in S, (v,u',\tau',d)\in E}\{\tau: \tau+d'\leq \pi_1(u')\}$. In particular, as Traveller is at $v$ at $\tau\geq \tau^*$ with $\tau^*>\pi_1(v^*)\geq \pi_1(v)$, when using the edge $(u,v,\tau,d)$ they reach $u$ strictly after $\pi_1(u)$. By induction hypothesis, Traveller has no winning strategy.  

\end{proof}

\section{Remaining part of the proof of Theorem~\ref{th:static}}\label{app:statichardness}

Now, let us first consider that the formula is satisfiable, and let $\sigma$ be an assignment satisfying the formula. We consider the following strategy for Traveller: first, they go from $s=v_0$ to $v_n$ with the path $(v_i,x_i,v_{i+1})$ if $x_i$ is true in $\sigma$ ($(v_i,\overline{x_i},v_{i+1})$ if $x_i$ is false).

Traveller reaches $v_n$ at time $2n$. At $v_n$, if Traveller has seen at least one blocked edge (i.e., either a blocked edge $\EDGE{x_i}{z_i}$, or $\EDGE{\overline{x_i}}{\overline{z_i}}$, or $\EDGE{v_n}{c_1}$), they go to $w$ and then to $t$ (there are four copies of $\EDGE{w}{t}$, and at most three further edges can be blocked), reaching $t$ at $2n+27M+2m=T$.

Otherwise, Traveller has not seen any blocked edge so far. They go to $c_1$ and continue their journey: at $c_j$, Traveller chooses a literal $\ell_j^s$ of $C_j$ satisfied by $\sigma$. It goes to $\alpha_j^s$. If at least one of the three copies of $\EDGE{\alpha_j^s}{c_{j+1}}$ is not blocked, then Traveller goes to $c_{j+1}$ and continues. If this happens for all clauses, Traveller reaches $c_{m+1}$ at time $2n+9M+2m$, and then $t$ at $2n+2m+27M=T$. 

Otherwise, all three edges $\EDGE{\alpha_j^s}{c_{j+1}}$ are blocked for some $j$. Then, at most, one further edge can be blocked. Traveller is at $\alpha_j^s$ at time $2n+9M+2j-1$. They go to $\beta_j^s$, then to the literal $x_i$ or $\overline{x_i}$ corresponding to $\ell_j^s$. They are here at time $2n+2m+17M$. Traveller knows that the edge $\EDGE{x_i}{z_i}$ (or $\EDGE{\overline{x_i}}{\overline{z_i}}$) is not blocked, as they have visited $x_i$ (or $\overline{x_i}$) when visiting vertex-gadgets. So Traveller goes to $z_i$ (or $\overline{z_i}$). Since one more edge can be blocked at most, Traveller can reach $t$, at time $2n+2m+27M=T$.\\

Now, suppose that the formula is not satisfiable. To reach $t$, Traveller must use at some point an edge among $\EDGE{x_i}{z_i}$, $\EDGE{\overline{x_i}}{\overline{z_i}}$, $\EDGE{x_i}{\beta_j^s}$ (or $\EDGE{\overline{x_i}}{\beta_j^s}$), $\EDGE{v_n}{c_1}$ or $\EDGE{v_n}{w}$.

{Blocker does} not block any edge until Traveller travels one of those edges. We look at the first time Traveller takes one of those edges.
\begin{itemize}
    \item If it is an edge $\EDGE{x_i}{z_i}$ or $\EDGE{\overline{x_i}}{\overline{z_i}}$, {Blocker blocks} the two edges to $t$. Then Traveller has to go back, with already at least $20M$. They need at least $10M$ to reach $t$ (no shorter path), so they arrive at $t$ at $\geq 30M >T$.
    \item If it is edge $\EDGE{v_n}{w}$, {Blocker blocks} the $4$ edges $\EDGE{w}{t}$, and Traveller has to go back and has already travelled $54M>T$.
    \item If it is an edge going to $\beta_j^s$, {Blocker blocks} the 2 edges $\EDGE{\alpha_j^s}{\beta_j^s}$. Traveller has to go back, having already travelled $16M$. The unique way to reach $t$ before $T$ is to go at some point to some $z_i$ or $\overline{z_i}$. {Blocker still has} 2 edges to block, so {they} block the two edges to $t$, and Traveller has to go back, having travelled at least $36M>T$.
    \item The only remaining case is that Traveller first uses $\EDGE{v_n}{c_1}$. Note that Traveller cannot use this edge again: indeed, they would have travelled at least $18M$ and would need at least $10M$ to reach $t$, thus at least $28M>T$. Traveller cannot be at $c_1$ before $2n+9M$. Note that the shortest path between $c_1$ and $t$ is $2m+18M$: either going to $c_{m+1}$ with a path of length $2m$ and then $t$, or to go to some $\alpha_j^s$ (length $2j-1$), then $\beta_j^s$, $x_i$ or $\overline{x_i}$, $z_i$ or $\overline{z_i}$ and $t$.
    
If Traveller is at $c_1$ strictly after $2n+9M$, they cannot reach $t$ by $T$. So Traveller must be at $c_1$ at time $2n+9M$. They took a shortest path from $s$ to $v_n$, visiting exactly one vertex among $x_i$ and $\overline{x_i}$, thus describing a truth value $\sigma$. To reach $t$ on time, Traveller must traverse some clause gadgets (using a shortest path). While they go to $\alpha_j^s$ where $\ell_j^s$ is true in $\sigma$, {Blocker does} not block any edge: if Traveller goes to $\beta_j^s$, they must go then to the corresponding $x_i$ or $\overline{x_i}$, then directly to $z_i$ or $\overline{z_i}$. {Blocker blocks} the two edges to $t$, and Traveller cannot reach $t$ on time.

     Then, Traveller must follow a (shortest) path on the clause-gadget. As the formula is not satisfied by $\sigma$, for some (unsatisfied) clause, Traveller will go to some $\ell_{j}^s$ corresponding to a false literal in $\sigma$. 
    {Blocker blocks} the $3$ edges $\EDGE{\ell_j^s}{c_{j+1}}$. Traveller must use a shortest path to $t$ (to be on time), so they must go to $\beta_j^s$, then to the false literal $x_i$ or $\overline{x_i}$. This vertex has not been visited yet by Traveller, so {Blocker} can block $\EDGE{x_i}{z_i}$ (or $\EDGE{\overline{x_i}}{\overline{z_i}}$). Traveller must then spend at least $2+10M$ to reach $t$, being late on $t$. 
\end{itemize}

\section{Proof of Theorem~\ref{th:utctp2}}\label{app:utctp2}

We need to make two kinds of changes to the graph used for Theorem~\ref{th:static}: the multiplicity of some edges and their time.

We choose them increasingly to ensure that any path allowed for construction can be taken. The multiplicity is adapted to $2$ edges. More precisely, the modifications are:
\begin{itemize}
    \item We put time $2i+1$ on edges $\EDGE{v_{i}}{x_{i+1}}$ and $\EDGE{v_{i}}{\overline{x_{i+1}}}$.  We put time $2i+2$ on edges $\EDGE{x_{i+1}}{v_{i+1}}$ and $\EDGE{\overline{x_{i+1}}}{v_{i+1}}$. This ensures that to go to $v_n$,  Traveller can and must see exactly one assignment per variable.
    \item $\EDGE{v_n}{w}$ gets time $2n+1$, and $\EDGE{w}{t}$ gets time $2n+2$. We put only two copies of $\EDGE{w}{t}$,  ensuring that this path to the target can be taken if and only if an edge is missing while visiting the variable gadgets.
    \item $\EDGE{v_n}{c_1}$ gets time $2n+1$.
    \item We put time $2n+2j$ on edges $\EDGE{c_j}{\alpha^s_j}$ and time $2n+2j+1$ on edges $\EDGE{\alpha^s_j}{c_{j+1}}$. We put a single copy of $\EDGE{\alpha^s_j}{c_{j+1}}$ instead of $3$. That way, Blocker only needs to block one edge to force Traveller to visit a variable they did not see earlier.
    \item Edge $\EDGE{c_{m+1}}{t}$ has time $L=2n+2m+2$, ensuring the target can be reached after seeing all clauses.
    \item Edge $\EDGE{\alpha^s_j}{\beta^s_j}$ is now forced and has time $L$. Edges $\EDGE{\beta^s_j}{ x_i}$ and $\EDGE{\beta^s_j}{ \overline{x_i}}$ get time $L+1$.
    \item Edges $\EDGE{x_i}{z_i}$ and $\EDGE{\overline{x_i}}{\overline{z_i}}$ get time $L+2$. Edges $\EDGE{\overline{z_i}}{t}$ get time $L+2$ (and we keep two copies). That way, Traveller cannot try to visit them before seeing a missing edge. The only reason Traveller goes to some $z_i$ or $\overline{z_i}$ is to have been blocked at some clause before.
\end{itemize}

When the formula is satisfiable, Traveller uses the same strategy as in the proof of Theorem~\ref{th:static}. Otherwise, Traveller must choose a single assignment as before. When they reach a literal of a clause not satisfied (i.e. visit some vertex $\alpha^s_j$), Blocker blocks the access to the next clause. They must reach the corresponding $x_i$ (resp. $\overline{x_i}$) and Blocker blocks the access to $z_i$ (resp. $\overline{z_i}$), leaving Traveller stuck.

\end{document}